[1994/12/01]
\documentclass[manuscript, a4paper, reqno]{aomart}

\chardef\bslash=`\\ 

\newtheorem[{}\it]{thm}{Theorem}[section]
\newtheorem{cor}[thm]{Corollary}
\newtheorem{lem}[thm]{Lemma}

\newtheorem{prop}[thm]{Proposition}

\theoremstyle{remark} 

\usepackage{dsfont}
\usepackage{amssymb}

\usepackage{float}

\usepackage{graphicx}
\usepackage{moresize}

\usepackage{hyperref}
\usepackage[numeric]{amsrefs}

\usepackage[english]{babel}

\theoremstyle{definition}
\newtheorem{defn}{\textsc{Definition}}[section]
\newtheorem{rem}{Remark}[section]

\newtheorem*[{}\it]{notation}{Notation}

\newtheorem*[{}\it]{rest}{\textsc{Theorem}}
\newtheorem*[{}\it]{quest}{\textsc{Question}}
\newtheorem*[{}\it]{projone}{\texttt{Project 1 (Future Work)}}
\newtheorem*[{}\it]{projtwo}{\texttt{Project 2 (Current Work)}}
\newtheorem*[{}\it]{projthree}{\texttt{Project 3 (Current Work)}}
\newtheorem*[{}\it]{projfour}{\texttt{Project 4 (Future Work)}}
\newtheorem*[{}\it]{projfive}{\texttt{Project 5 (Future Work)}}
\newtheorem*[{}\it]{proofoflemma}{Proof of Lemma}

\usepackage[textwidth=6in,textheight=9.5in, paperwidth=7.5in,paperheight=11.25in, inner=2cm, outer=2cm]{geometry} 

\usepackage[usenames,dvipsnames]{xcolor}

\usepackage{txfonts}

\usepackage{multicol}

\usepackage{enumerate}

\title[]{On Collision Invariants for Linear Scattering}
\author[L. Saint-Raymond]{Laure Saint-Raymond}\thanks{L. Saint-Raymond, {\em Department of Mathematics, Harvard University, and} \\ {\em D\'{e}partement de Math\'{e}matiques et Applications, \'{E}cole Normale Sup\'{e}rieure, Paris} \href{mailto:laure@math.harvard.edu}{(laure@math.harvard.edu)}}
\author[M. Wilkinson]{Mark Wilkinson}\thanks{M. Wilkinson, {\em Courant Institute of Mathematical Sciences, New York University} \href{mailto:mwilkins@cims.nyu.edu}{(mwilkins@cims.nyu.edu)}}

\allowdisplaybreaks

\newcommand{\ov}{\overline}

\newcommand{\tempb}{\multicolumn{1}{c|}{I_{4}-2\widehat{\gamma}(\psi)\otimes \widehat{\gamma}(\psi)}}
\newcommand{\tempd}{\multicolumn{1}{|c}{I_{2}}}

\usepackage{tikz}
\usetikzlibrary{patterns}
\newcommand{\boundellipse}[3]
{(#1) ellipse (#2 and #3)
}

\begin{document}

\maketitle

\begin{abstract}
\noindent In this article, we extend the result of \textsc{Boltzmann} \cite{boltzmann2012wissenschaftliche} on characterisation of collision invariants from the case of hard disks to a class of two-dimensional compact, strictly-convex particles.
\end{abstract}
%

\vspace{3mm}

\section{Introduction}

Understanding the statistical behaviour of dynamical systems comprised of identical interacting particles has been a well-studied problem since the work of \textsc{Boltzmann} \cite{boltzmann2012wissenschaftliche} in the kinetic theory of rarified gases. By studying the precise way in which particles scatter from each other following a collision, one is able to derive information about {\em macroscopic} properties of the system, such as the evolution of the local density of the gas or local propagation of heat. A great portion of the kinetic theory literature is devoted to the study of systems in which the identical particles are perfect spheres. However, it is a very natural question to understand in what ways the statistical properties of systems of non-spherical particles differ from those composed of their perfectly spherical counterparts.

In this article, we offer a preliminary contribution to the extension of the theory of the Boltzmann equation from hard spheres to general hard particles. In the first part of this work, we study the physical dynamics of compact, strictly-convex bodies which do not interpenetrate. Moreover, we restrict our attention to systems of two identical particles, thereby considering binary particle interactions alone. The first important step in studying such systems is to construct suitable physical boundary conditions for a dynamics (by means of {\em scattering maps}) when the two hard particles collide, in order that trajectories in phase space may be defined globally in time. By `physical' boundary conditions, we mean that (i) the particles should not interpenetrate following collision, and (ii) there should also be conservation of total linear momentum, angular momentum and kinetic energy of the two particles through any collision event. However, it is important to note here that, according to \textsc{Wilkinson} \cite{wilk}, it is not possible to construct a family of scattering matrices corresponding to the collision of two non-spherical particles which conserves their total linear momentum, angular momentum and kinetic energy. Nevertheless, with the extension of Boltzmann's equation to systems of non-spherical particles in mind, we construct and study families of scattering matrices for two particle systems which conserve total linear momentum and kinetic energy of the colliding particles.

The second and principal part of this paper is devoted to the important topic of {\em collision invariants} for non-spherical particle scattering in kinetic theory. To illustrate the importance of collision invariants, we turn very briefly to the case of hard particles with spherical symmetry in $\mathbb{R}^{d}$ (which are hard disks in the case $d=2$, but hard spheres in the case $d=3$) and the classical Boltzmann equation.
\subsection{The Boltzmann Equation and Collision Invariants}
It is well known that the Boltzmann equation for the 1-particle density function $f=f(x, v, t)$ given by
\begin{equation}\label{boltzmanneq}
\frac{\partial f}{\partial t}+(v\cdot \nabla_{x})f=\mathcal{C}(f, f) \quad \text{for}\hspace{2mm} (x, v)\in\mathbb{R}^{d}\times\mathbb{R}^{d}
\end{equation}
is a candidate PDE to describe the statistical properties of systems of $N$ hard particles with spherical symmetry in the Boltzmann-Grad limit as $N\rightarrow \infty$ and $\varepsilon\rightarrow 0$ with $N\varepsilon^{d-1}=1$, where $\varepsilon>0$ denotes the radius of any given particle. The unique family of scattering matrices $\{\sigma_{n}\}_{n\in\mathbb{S}^{d-1}}$ which resolves a collision between two spherical particles, in such a way that properties (i) and (ii) above are satisfied, are the reflection matrices
\begin{equation*}
\sigma_{n}:=I-2\widehat{\gamma}_{n}\otimes\widehat{\gamma}_{n}\in\mathrm{O}(2d),
\end{equation*}
with $\widehat{\gamma}_{n}:=\frac{1}{\sqrt{2}}[n, -n]$, where $n\in\mathbb{S}^{d-1}$ denotes the direction connecting the centres of mass of the two spheres at collision. The {\em collision operator} $\mathcal{C}(f, f)$ that appears in \eqref{boltzmanneq} is given by
\begin{equation}\label{collisionoperator}
\mathcal{C}(f, f):=\frac{1}{2}\int_{\mathbb{R}^{d}}\int_{\mathbb{S}^{d-1}_{+}}|(v-\ov{v})\cdot n|\big(f(x, v_{n}', t)f(x, \ov{v}_{n}', t)-f(x, v, t)f(x, \ov{v}, t)\big)\,dn d\ov{v},
\end{equation} 
where the `post-collisional' velocities $[v_{n}', \ov{v}_{n}']:=\sigma_{n}[v, \ov{v}]\in\mathbb{R}^{2d}$ are
\begin{equation*}
v_{n}'=v-[(v-\ov{v})\cdot n]n \quad\text{and}\quad \ov{v}_{n}'=\ov{v}+[(v-\ov{v})\cdot n]n.
\end{equation*}

In order to derive laws for the local conservation of mass, linear momentum and kinetic energy associated to the Boltzmann equation, one must consider velocity averages of solutions of \eqref{boltzmanneq} with respect to an appropriate integrable function $\phi:\mathbb{R}^{d}\rightarrow\mathbb{R}$, and in turn use elementary properties of the family of Boltzmann scattering matrices $\{\sigma_{n}\}_{n\in\mathbb{S}^{d-1}}$. Indeed, one can show formally that
\begin{equation*}
\frac{\partial}{\partial t}\int_{\mathbb{R}^{d}}\phi f\,dv+\nabla_{x}\cdot\int_{\mathbb{R}^{d}}\phi fv\,dv = \frac{1}{4}\int_{\mathbb{R}^{d}}\mathcal{C}(f, f)(\phi(v)+\phi(\ov{v})-\phi(v_{n}')-\phi(\ov{v}_{n}'))\,dv,
\end{equation*}
whence
\begin{equation*}
\frac{\partial}{\partial t}\int_{\mathbb{R}^{d}}\phi f\,dv+\nabla_{x}\cdot\int_{\mathbb{R}^{d}}\phi fv\,dv=0
\end{equation*}
if $\phi$ satisfies the identity
\begin{equation}\label{fi}
\phi(v_{n}')+\phi(\ov{v}_{n}')=\phi(v)+\phi(\ov{v}),
\end{equation}
for all $V=[v, \ov{v}]\in\mathbb{R}^{2d}$ and $n\in\mathbb{S}^{d-1}$. By choosing $\phi=\phi(v)$ to be $1, v$ or $\frac{1}{2}|v|^{2}$, one recovers PDE expressing the local conservation of mass, linear momentum and kinetic energy for $f$, respectively. 

Another important observation in the theory of the Boltzmann equation is that the entropy map
\begin{equation*}
f\mapsto\int_{\mathbb{R}^{d}}\int_{\mathbb{R}^{d}}f\log{f}\,dxdv
\end{equation*}
is a formal Lyapunov functional for the dynamics generated by \eqref{boltzmanneq}, since it can be shown that
\begin{equation}\label{entrop}
\int_{\mathbb{R}^{d}}\mathcal{C}(f, f)\log{f}\,dv=-\frac{1}{4}\int_{\mathbb{R}^{d}}\mathcal{C}(f, f)\log\left(\frac{f'\ov{f}'}{f\ov{f}}\right)\,dv\leq 0
\end{equation}
with equality holding if and only if $f$ is a Maxwellian distribution $f_{\mathrm{M}}$,
\begin{equation*}
f_{\mathrm{M}}(v)=\frac{\rho}{(2\pi\Theta)^{d/2}}\exp\left(-\frac{|v-u|^{2}}{2\Theta}\right) \quad \text{for some}\hspace{2mm} \rho, \Theta>0 \hspace{2mm} \text{and}\hspace{2mm} u\in\mathbb{R}^{d}.
\end{equation*}
In order to demonstrate that all minimisers of the entropy functional \eqref{entrop} (amongst a suitable class of admissible functions) are indeed Maxwellia, one also needs to characterise all solutions $\phi:\mathbb{R}^{d}\rightarrow\mathbb{R}$ of the functional equation \eqref{fi}. Knowledge of all collision invariants also provides us with the nullspace of $\mathcal{L}_{f_{M}}$, the linearisation of the collision operator \eqref{collisionoperator} about a global Maxwellian $f_{M}$, which is crucial when it comes to investigating the behaviour of perturbations of equilibrium solutions of the Boltzmann equation \eqref{boltzmanneq}. Moreover, characterisation of collision invariants is important for establishing rigorous connections between the Boltzmann kinetic equation and the Euler and Navier-Stokes equations of fluid dynamics: see \textsc{Bardos, Golse and Levermore} \cite{MR1115587, MR1213991} for more on such ideas. 

Under various assumptions on $\phi$, it has been shown in the work of many authors (for instance \textsc{Boltzmann} \cite{boltzmann2012wissenschaftliche} for the $C^{1}$ case, \textsc{Gr\"{o}nwall} \cite{MR1503514} for the $C^{0}$ case, \textsc{Cercignani} \cite{MR1049048} for the Maxwellian-weighted $L^{2}$ case, and \textsc{Arkeryd} \cite{MR0339666} for the $L^{1}_{\mathrm{loc}}$ case) that if a scalar function $\phi:\mathbb{R}^{d}\rightarrow\mathbb{R}$ satisfies $\phi(v_{n}')+\phi(\ov{v}_{n}')=\phi(v)+\phi(\ov{v})$ for all $V=[v, \ov{v}]\in\mathbb{R}^{2d}$ and $n\in\mathbb{S}^{d-1}$, it is necessarily of the form
\begin{equation*}
\phi(v)=a+b\cdot v+c|v|^{2},
\end{equation*}
for some constants $a, b_{1}, ..., b_{d}, c\in\mathbb{R}$. Any such function $\phi$ is known as a collision invariant, as the value of the map $[v, \ov{v}]\mapsto\phi(v)+\phi(\ov{v})$ does not change when `pre-collisional velocities' are changed to their `post-collisional' values by $\sigma_{n}$ for any $n\in\mathbb{S}^{d-1}$. In this article, we will focus our efforts on establishing the analogue of this result when the particles in the underlying dynamical system are no longer perfectly spherical.

Although the motivation for studying collision invariants can be found at the kinetic level, we make no further study of the Boltzmann equation in the sequel. In all that follows, we focus our attention solely at the level of particles.
\subsection{Informal Statements of Main Results}
As it takes quite some effort to set up precise statements of the main results of this article, we state them at first in a somewhat informal manner. For simplicity, we work in two spatial dimensions in all the sequel, i.e. we consider the motion of two-dimensional particles evolving in the whole space $\mathbb{R}^{2}$. However, all results in this article can be extended to the case of three-dimensional particles evolving in the whole space $\mathbb{R}^{3}$.

We study the dynamics of systems of non-spherical particles $\mathsf{P}$ consisting of two identical compact, strictly-convex subsets with analytic boundaries, i.e. $\partial\mathsf{P}$ is of class $C^{\omega}$. Naturally, we stipulate that at no time should the particles interpenetrate. As such, we must construct a dynamics on a suitable phase space of {\em hard particles} (see section \ref{derivationofmodel} below for the precise definition of `hard particle phase space'). The dynamics of the hard particles is governed by Euler's Laws of Motion, the analogue of Newton's Laws for continuum rigid bodies. The first result of this article concerns the existence of solutions to Euler's equations for their evolution which conserve the total linear momentum and kinetic energy of initial data for all time, and which also ensure non-penetration of the particles for all time. Informally stated, we establish the following result:
\begin{thm}\label{poorflow}
Consider two identical compact, strictly-convex particles with analytic boundary. There exist global-in-time classical solutions to Euler's equations of motion on the phase space of all particle configurations for which there is no particle interpenetration. Moreover, these classical solutions conserve the total linear momentum and kinetic energy of any given initial datum for all time.
\end{thm}
The precise version of \textsc{Theorem} \ref{poorflow} is stated as \textsc{Theorem} \ref{flow} below. The proof of this result makes use of the general existence theory of \textsc{Ballard} \cite{ballard} for dynamics of rigid bodies with non-penetration constraints. However, in order to invoke his theory one must first construct scattering matrices which resolve collisions between two compact, strictly-convex sets in such a way that total linear momentum and kinetic energy are conserved. This construction is performed in section \ref{scatteringmaps} below. The reader might notice that the statement of \textsc{Theorem} \ref{poorflow} does not claim that total {\em angular momentum} is conserved for all time by classical solutions (whose precise definition is given in \ref{classicalsolutions} below). In fact, it has been shown in \textsc{Wilkinson} \cite{wilk} that classical solutions of Euler's equations which conserve total linear momentum, angular momentum and kinetic energy of initial data for all time do not exist for all possible initial data. It is for this reason we confine our attention in this article to dynamics which conserve only linear momentum and kinetic energy, since the notion of scattering map and classical solutions to Euler's equations of motion are intimately related to one another. Let us also draw attention to the fact that it may, at first glance, seem that our choice of dynamics is somewhat arbitrary, since one can construct distinct families of solution operators $\{T_{t}\}_{t\in\mathbb{R}}$ associated to Euler's equations which conserve total linear momentum and kinetic energy for all time. We justify our particular choice of dynamics $\{T_{t}\}_{t\in\mathbb{R}}$ in section \ref{compare} below.

While the spatial collision configuration of two hard disks can be characterised by the single angle that the line connecting their centres of mass makes with a given reference line, we note that an element $\beta$ of the three-torus $\mathbb{T}^{3}$ is required to characterise the spatial collision configuration of two compact, strictly-convex particles which are not disks. To see this, one might wish to consult figure 2 below. With this in mind, we present an informal statement of the main result of this article.

\begin{thm}\label{poorawesome}
Suppose a measurable map $\varphi:\mathbb{R}^{2}\times\mathbb{R}\times\mathbb{S}^{1}\rightarrow \mathbb{R}$ satisfies the functional identity for collision invariants given by 
\begin{equation*}
\varphi(v_{\beta}', \omega_{\beta}', \vartheta)+\varphi(\ov{v}_{\beta}', \ov{\omega}_{\beta}', \ov{\vartheta})=\varphi(v, \omega, \vartheta)+\varphi(\ov{v}, \ov{\omega}, \ov{\vartheta})
\end{equation*}
for every $V=[v, \ov{v}, \omega, \ov{\omega}]\in\mathbb{R}^{6}$ and all $\beta\in\mathbb{T}^{3}$, where $[v_{\beta}', \ov{v}_{\beta}', \omega_{\beta}', \ov{\omega}_{\beta}']\in\mathbb{R}^{6}$ denotes the post-collisional values of the vector $V$ corresponding to the spatial configuration $\beta$. Then $\varphi$ is necessarily of the form
\begin{equation*}
\varphi(v, \omega, \vartheta):=a(\vartheta)+b\cdot v+c\left(m|v|^{2}+J\omega^{2}\right),
\end{equation*}
for some constants $b_{1}, b_{2}, c\in\mathbb{R}$ and some measurable function $a:\mathbb{S}^{1}\rightarrow\mathbb{R}$.
\end{thm}
The precise statement of this result appears as \textsc{Theorem} \ref{awesome} below.
\subsection{Structure of the Article}
In section \ref{hardspherechar}, we revisit the case of hard disk scattering and present a new proof of characterisation collision invariants. We derive the equations of motion for the physical evolution of hard particles in section \ref{derivationofmodel}. The concept of scattering map and regularity of solutions of Euler's equations are intimately linked, so in sections \ref{scatteringmaps} and \ref{connydym} we construct families of scattering maps and, in turn, classical solutions to Euler's equations of motion. In the final part of the paper, namely section \ref{mainsection}, we characterise collision invariants for compact, strictly-convex non-spherical particles. New results by C. Viterbo on generators of orthogonal groups of matrices, which allow us to establish the proof of \textsc{Theorem} \ref{poorawesome}, are stored in the appendix \ref{algebra}.

\section{Characterisation of Collision Invariants for Hard Disks: A New and Simple Method}\label{hardspherechar}
Before we embark upon the problem of characterising collision invariants for general convex particle scattering maps, it will be helpful to recall the theory which has been established in the case of spherical particles (or, more appropriately in our two-dimensional setting, particles which are disks). Our approach to this problem appears to be new, and has the advantage of requiring no regularity or integrability conditions on the collision invariant $\phi$, only that it be measurable. Although we only discuss scattering of hard disks in $\mathbb{R}^{2}$ in this section, all our results also hold for the scattering of hard spheres in $\mathbb{R}^{3}$.
\subsection{State-of-the-art of Previously-established Results}
For any $\psi\in\mathbb{S}^{1}$, consider the associated Boltzmann scattering map $\sigma_{\psi}:\mathbb{R}^{4}\rightarrow\mathbb{R}^{4}$ for two hard disks given by
\begin{equation}\label{boltzscat}
\sigma_{\psi}[V]:=\left(I-2\widehat{\gamma}_{\psi}\otimes\widehat{\gamma}_{\psi}\right)V \quad \text{with}\hspace{2mm} V=[v, \ov{v}]\in\mathbb{R}^{4},
\end{equation}
where 
\begin{displaymath}
\widehat{\gamma}_{\psi}:=\frac{1}{\sqrt{2}}\left[
\begin{array}{c}
e(\psi) \\ -e(\psi)
\end{array}
\right],
\end{displaymath}
with $e(\psi):=(\cos\psi, \sin\psi)\in\mathbb{R}^{2}$ and $\psi$ denotes the angle that the line connecting the centres of mass of the colliding disks makes with the positive $x$-axis. One can check that for every choice of $\psi\in\mathbb{S}^{1}$, the scattering map $\sigma_{\psi}$ conserves total linear momentum, angular momentum and kinetic energy of any given velocity vector $V\in\mathbb{R}^{4}$. Under the assumptions that $\phi:\mathbb{R}^{2}\rightarrow\mathbb{R}$ be in $L^{1}_{\mathrm{loc}}(\mathbb{R}^{2})$ and satisfy the functional equation
\begin{equation}\label{diskcollinv}
\phi(v_{\psi}')+\phi(\ov{v}_{\psi}')=\phi(v)+\phi(\ov{v})
\end{equation}
pointwise almost everywhere on $\mathbb{R}^{4}\times\mathbb{S}^{1}$, where the post-collisional velocities $v_{\psi}'$ and $\ov{v}_{\psi}'$ are given in terms of $\sigma_{\psi}[V]=(\sigma_{\psi}[V]_{1}, ..., \sigma_{\psi}[V]_{4})$ as
\begin{displaymath}
v_{\psi}':=\left(
\begin{array}{c}
\sigma_{\psi}[V]_{1}\\
\sigma_{\psi}[V]_{2}
\end{array}
\right)\quad\text{and}\quad \ov{v}_{\psi}':=\left(
\begin{array}{c}
\sigma_{\psi}[V]_{3} \\
\sigma_{\psi}[V]_{4}
\end{array}
\right),
\end{displaymath}
it has been shown by \textsc{Arkeryd} (\cite{MR0339666}, lemma 2.8) that $\phi$ is necessarily of the form $\phi(v)=a+b\cdot v+c|v|^{2}$ almost everywhere for some constants $a, b_{1}, b_{2}, c\in\mathbb{R}^{2}$. Our new proof of characterisation of collision invariants covers the case where $\phi$ is only measurable on $\mathbb{R}^{2}$, as opposed to being of class $L^{1}_{\mathrm{loc}}(\mathbb{R}^{2})$. On the other hand, we ask that the identity \eqref{diskcollinv} hold for all $\psi\in\mathbb{S}^{1}$ and for all $V\in\mathbb{R}^{4}$. In order to produce the most general result possible, one would need to extend our argument to the case where \eqref{diskcollinv} holds for almost every $\psi\in\mathbb{S}^{1}$ and almost every $V\in\mathbb{R}^{4}$, as opposed to everywhere on $\mathbb{S}^{1}$ and $\mathbb{R}^{4}$, respectively. We do not attempt do this here.
\subsection{Orbits of Scattering Groups on $\mathbb{R}^{4}$}
In order to motivate our new group-theoretic approach in the case of general strictly-convex particles, let us rewrite identity \eqref{diskcollinv} as 
\begin{equation}\label{goody}
\Phi_{\phi}(\sigma_{\psi}[V])=\Phi_{\phi}(V) 
\end{equation}
for $V\in\mathbb{R}^{4}$ and $\psi\in\mathbb{S}^{1}$, where
\begin{equation*}
\Phi_{\phi}(V):=\phi(v)+\phi(\ov{v}),
\end{equation*}
with $V=[v, \ov{v}]\in\mathbb{R}^{4}$, assuming that $\phi$ be only measurable and, without loss of generality, that $\phi(0)=0$ and thus $\Phi_{\phi}(0)=0$. In particular, identity \eqref{goody} implies that for any fixed choice of $V$ and any collection of angles $\psi_{1}, ..., \psi_{k}\in\mathbb{S}^{1}$, one has
\begin{equation*}
\Phi_{\phi}\left(\sigma_{\psi_{k}}\circ ... \circ \sigma_{\psi_{1}}[V]\right)=\Phi_{\phi}(V),
\end{equation*}
namely that the map $\Phi_{\phi}$ is constant on the left group orbits $GV\subset\mathbb{R}^{4}$ for any given $V\in\mathbb{R}^{4}$, where $G\subseteq\mathrm{O}(4)$ is the group generated by the 1-parameter family of reflection matrices $\{I-2\widehat{\gamma}_{\psi}\otimes\widehat{\gamma}_{\psi}\,:\,\psi\in\mathbb{S}^{1}\}$, namely
\begin{equation}\label{scatgroup}
G:=\left\langle\left\{I-2\widehat{\gamma}_{\psi}\otimes\widehat{\gamma}_{\psi}\,:\,\psi\in\mathbb{S}^{1}\right\}\right\rangle.
\end{equation}
Let us now find the group orbits $GV$ for any $V\in\mathbb{R}^{4}$. For $\mathsf{e}>0$ and $\mathsf{p}\in\mathbb{R}^{2}$ satisfying $\mathsf{e}^{2}>|\mathsf{p}|^{2}/2$, we define $\mathsf{M}(\mathsf{e}, \mathsf{p})$ to be the subset of $\mathbb{R}^{4}$ given by
\begin{displaymath}
\mathsf{M}(\mathsf{e}, \mathsf{p}):=\left\{Y\in\mathbb{R}^{4}\,:\,|Y|^{2}=\mathsf{e}^{2}\hspace{2mm}\text{and}\hspace{2mm}\left(
\begin{array}{c}
Y_{1}+Y_{3} \\ Y_{2}+Y_{4}
\end{array}\right)=\mathsf{p}\right\},
\end{displaymath}
which is evidently homeomorphic to $\mathbb{S}^{1}$. When $\mathsf{e}^{2}=|\mathsf{p}|^{2}/2$, $\mathsf{M}(\mathsf{e}, \mathsf{p})$ is a singleton and when $\mathsf{e}^{2}<|\mathsf{p}|^{2}/2$, one can check $\mathsf{M}(\mathsf{e}, \mathsf{p})$ is empty. It is clear that when $V\in\mathbb{R}^{4}$ is given, the Boltzmann scattering matrix $\sigma_{\psi}$ maps $\mathsf{M}(\mathsf{e}, \mathsf{p})$ to itself for any $\psi\in\mathbb{S}^{1}$, where $\mathsf{e}=|V|$ and $\mathsf{p}=(V_{1}+V_{3}, V_{2}+V_{4})$. 
\subsection{Reduction to Canonical Form}\label{diskcanon}
As the sets $\mathsf{M}(\mathsf{e}, \mathsf{p})$ are homeomorphic to $\mathbb{S}^{1}$ for $\mathsf{e}^{2}>|\mathsf{p}|^{2}$, we can expect to reduce our study of scattering groups acting on $\mathsf{M}(\mathsf{e}, \mathsf{p})$ to the study of some other group acting on $\mathbb{S}^{1}$. To show this, we reduce our problem to a kind of canonical form. Indeed, for $\mathsf{e}^{2}>|\mathsf{p}|^{2}/2$, we consider the bijection $h_{\mathsf{e}, \mathsf{p}}:\mathsf{M}(\mathsf{e}, \mathsf{p})\rightarrow \mathbb{S}^{1}$ given by
\begin{displaymath}
h_{\mathsf{e}, \mathsf{p}}[V]:=\frac{1}{\sqrt{(V_{1}-V_{3})^{2}+(V_{2}-V_{4})^{2}}}\left(
\begin{array}{cc}
V_{1}-V_{3} \\
V_{2}-V_{4}
\end{array}
\right) \quad \text{for}\hspace{2mm} V\in\mathsf{M}(\mathsf{e}, \mathsf{p}),
\end{displaymath}
with inverse given by
\begin{displaymath}
h_{\mathsf{p}, \mathsf{e}}^{-1}[\zeta]:=\frac{1}{2}\left(
\begin{array}{c}
\sqrt{2\mathsf{e}^{2}-|\mathsf{p}|^{2}}\zeta_{1}+\mathsf{p}_{1} \\
\sqrt{2\mathsf{e}^{2}-|\mathsf{p}|^{2}}\zeta_{2}+\mathsf{p}_{2} \\
\mathsf{p}_{1}-\sqrt{2\mathsf{e}^{2}-|\mathsf{p}|^{2}}\zeta_{1} \\
\mathsf{p}_{2}-\sqrt{2\mathsf{e}^{2}-|\mathsf{p}|^{2}}\zeta_{2}
\end{array}
\right)\quad \text{for}\hspace{2mm}\zeta=(\zeta_{1}, \zeta_{2})\in\mathbb{S}^{1}.
\end{displaymath}
One has that $\sigma_{\psi}\in\mathbb{R}^{4\times 4}$ maps $V$ to $(I-2\widehat{\gamma}_{\psi}\otimes\widehat{\gamma}_{\psi})V$ if and only if the matrix 
\begin{equation*}
s_{\psi}:=I-2e(\psi)\otimes e(\psi)\in\mathbb{R}^{2\times 2}
\end{equation*}
maps $h_{\mathsf{e}, \mathsf{p}}[V]$ to $(I-2e(\psi)\otimes e(\psi))h_{\mathsf{e}, \mathsf{p}}[V]$. Thus, if the group $\langle \{s_{\psi}\,:\,\psi\in\mathbb{S}^{1}\}\rangle\subseteq \mathrm{O}(2)$ acts transitively on the circle $\mathbb{S}^{1}$, it will follow immediately that the group orbit $GV$ is identically equal to $\mathsf{M}(\mathsf{e}, \mathsf{p})$. This is indeed the case, as the following elementary result shows.
\begin{prop}
The group $\langle\{I-2e(\psi)\otimes e(\psi)\,:\,\psi\in\mathbb{S}^{1}\}\rangle\subseteq\mathrm{O}(2)$ acts transitively on $\mathbb{S}^{1}$.
\end{prop}
\begin{proof}
For any two points $\zeta_{1}=e(\psi_{1})$ and $\zeta_{2}=e(\psi_{2})$ for $\psi_{1}, \psi_{2}\in\mathbb{S}^{1}$, we set $\psi':=(\psi_{1}+\psi_{2})/2\in\mathbb{S}^{1}$. One can check that $\zeta_{2}=(I-2e(\psi')^{\perp}\otimes e(\psi')^{\perp})\zeta_{1}$, and so we are done.
\end{proof}
Transforming back to $\mathbb{R}^{4}$, we immediately infer that the orbits of points of $\mathbb{R}^{4}$ under the action of the scattering group $G$ in \eqref{scatgroup} above are given by
\begin{displaymath}
GV=\left\{
\begin{array}{ll}
\mathsf{M}(\mathsf{e}, \mathsf{p}) & \quad \text{if} \hspace{2mm} \mathsf{e}^{2}>\frac{|\mathsf{p}|^{2}}{2}\vspace{2mm}\\
\left\{[\frac{1}{2}\mathsf{p}, \frac{1}{2}\mathsf{p}]\right\} & \quad \text{if}\hspace{2mm} \mathsf{e}^{2}=\frac{|\mathsf{p}|^{2}}{2}
\end{array}
\right.
\end{displaymath}
Since $\Phi_{\phi}$ is constant on each left orbit $GV$, it follows that
\begin{equation*}
\Phi_{\phi}(V)=\widetilde{\Phi_{\phi}}(v+\ov{v}, |v|^{2}+|\ov{v}|^{2})
\end{equation*}
for some new measurable function $\widetilde{\Phi_{\phi}}:\mathbb{R}^{2}\times\mathbb{R}\rightarrow\mathbb{R}$. One may then check (using the fact that $\phi(0)=0$) that $\widetilde{\Phi_{\phi}}$ satisfies the identity
\begin{equation}\label{usemecauchy}
\widetilde{\Phi_{\phi}}(v, |v|^{2})+\widetilde{\Phi_{\phi}}(\ov{v}, |\ov{v}|^{2})=\widetilde{\Phi_{\phi}}(v+\ov{v}, |v|^{2}+|\ov{v}|^{2}) \quad \text{for all}\hspace{2mm}V=[v, \ov{v}]\in\mathbb{R}^{4}.
\end{equation}
It is at this point we appeal to results on the characterisation of solutions to Cauchy's Functional Equation (see, for instance, the book of \textsc{Kuczma} \cite{MR2467621}). 
\subsection{Results on Cauchy's Functional Equation}
We recall that, under the assumption $f:\mathbb{R}^{2}\rightarrow \mathbb{R}$ be a measurable function, any solution of the functional identity
\begin{equation}\label{cauchy}
f(x)+f(y)=f(x+y)\quad \text{for all}\hspace{2mm}x, y\in\mathbb{R}^{2}
\end{equation}
is necessarily of the form $f(x)=cx$ for some $c\in\mathbb{R}$. We remark in passing that one cannot weaken the assumption that $\phi$ is measurable, if one wishes to avoid dealing with `pathological' solutions of Cauchy's functional equation. Indeed, by dropping the assumption of measurability and assuming the axiom of choice, it has been shown by \textsc{Hamel} \cite{hamel} that there exist discontinuous solutions of \eqref{cauchy}. 

One can use the fact that all measurable solutions of \eqref{cauchy} are of the form $f(x)=cx$ to characterise all measurable maps satisfying the functional equation \eqref{usemecauchy} for $\widetilde{\Phi_{\phi}}$ above. We now quote a result contained in \textsc{Truesdell and Muncaster} (\cite{MR554086}, pages 72--73 and pages 88--89), whose proof we revisit in detail in section \ref{itsdone}. 
\begin{prop}
Suppose that a measurable map $\Phi:\mathbb{R}^{2}\times\mathbb{R}\rightarrow\mathbb{R}$ satisfies the identity
\begin{equation}\label{funcid}
\Phi(v, |v|^{2})+\Phi(\ov{v}, |\ov{v}|^{2})=\Phi(v+\ov{v}, |v|^{2}+|\ov{v}|^{2})
\end{equation}
for all $v, \ov{v}\in\mathbb{R}^{2}$. It follows that $\Phi$ is necessarily of the form $\Phi(v, |v|^{2})=b\cdot v + c|v|^{2}$ for some constants $b_{1}, b_{2}, c\in\mathbb{R}$.
\end{prop}
%
Using the additional observation that any constant function is also a collision invariant, it quickly follows that if a measurable function $\phi:\mathbb{R}^{2}\rightarrow\mathbb{R}$ satisfies the identity
\begin{equation*}
\phi(v_{\psi}')+\phi(\ov{v}_{\psi}')=\phi(v)+\phi(\ov{v}) \quad \text{for all}\hspace{2mm}v\in\mathbb{R}^{3}\hspace{2mm}\text{and}\hspace{2mm}\psi\in\mathbb{S}^{1},
\end{equation*}
then it is necessarily of the form $\phi(v)=a+b\cdot v+c|v|^{2}$. As such, one can view the problem of characterisation of collision invariants as the problem of classifying all scalar invariants of a given group action (namely that of the scattering group $G$) on Euclidean space $\mathbb{R}^{4}$. It appears that this perspective on the problem is new. In particular, we emphasise that we placed only minimal assumptions on $\phi$, namely that it be only measurable on $\mathbb{R}^{2}$. It is this group-theoretic perspective on the problem we adopt in order to prove the main result of this article, namely \textsc{Theorem} \ref{poorawesome} (restated precisely as \textsc{Theorem} \ref{awesome} below). We now leave the case of hard disks to study general compact, strictly-convex sets with $C^{\omega}$ boundaries.

\section{Dynamics of Compact, Strictly-convex Particles}\label{derivationofmodel}
Although collision invariants themselves have no relationship to particle dynamics, what constitute pre- and post-collisional velocities at collision is, however, inherently a dynamical issue. It is for this reason we must address the dynamics of particles in this article. As collision invariants only involve two-particle interactions, we study in all the sequel the evolution of two compact, strictly convex sets $t\mapsto \mathsf{P}(t)$ and $t\mapsto\ov{\mathsf{P}}(t)$ in the plane $\mathbb{R}^{2}$ which do not interpenetrate. We assume that their boundary curves are of class $C^{\omega}$, and that the motion of $\mathsf{P}$ and $\ov{\mathsf{P}}$ takes place in the absence of external forces. We subsequently refer to compact, strictly-convex subsets of $\mathbb{R}^{2}$ as {\em hard particles}. As there are no externally-imposed forces in our systems under consideration, the evolution of the sets $\mathsf{P}(t)$ and $\ov{\mathsf{P}}(t)$ {\em before collision} is determined by their initial states, namely their initial spatial configurations (centres of mass and orientations) and initial velocities (both linear and angular). In order to construct a `physical' evolution for these two hard particles on $\mathbb{R}^{2}$, we appeal to {\em Euler's Laws of Motion} for continuum rigid body classical mechanics. We recall that Euler's laws are the appropriate extension of Newton's laws of motion to the study of continuum rigid bodies. We refer the reader to \textsc{Truesdell} (\cite{MR1162744}) for more on this topic.

Let us now set up the basic objects with which we work throughout this article. Suppose that $\mathsf{P}_{\ast}\subset\mathbb{R}^{2}$ is a compact, strictly-convex set with boundary of class $C^{\omega}$. Moreover, suppose that its centre of mass lies at the origin, i.e.
\begin{equation*}
\int_{\mathsf{P}_{\ast}}y\,dy = 0.
\end{equation*} 
We shall subsequently call any such set a {\em reference particle}. When an arbitrary centre of mass $x\in\mathbb{R}^{2}$ and orientation $\vartheta\in\mathbb{S}^{1}$ have been given, we write the $x$-translate and $\vartheta$-rotation of $\mathsf{P}_{\ast}$ as
\begin{equation*}
\mathsf{P}(x, \vartheta):=R(\vartheta)\mathsf{P}_{\ast}+x, 
\end{equation*}
where $R(\alpha)\in\mathrm{SO}(2)$ is the rotation matrix
\begin{displaymath}
R(\alpha):=\left(
\begin{array}{cc}
\cos\alpha & -\sin\alpha \\
\sin\alpha & \cos\alpha
\end{array}
\right).
\end{displaymath}
The evolution of the sets $\mathsf{P}(t)$ and $\ov{\mathsf{P}}(t)$ is expressed by
\begin{equation*}
\mathsf{P}(t):=R(\vartheta(t))\mathsf{P}_{\ast}+x(t)\quad\text{and}\quad \ov{\mathsf{P}}(t):=R(\ov{\vartheta}(t))\mathsf{P}_{\ast}+\ov{x}(t),
\end{equation*}
with the centres of mass $x(t), \ov{x}(t)\in\mathbb{R}^{2}$ and orientations $\vartheta(t), \ov{\vartheta}(t)\in\mathbb{S}^{1}$ being related to the linear velocities $v(t), \ov{v}(t)\in\mathbb{R}^{2}$ and angular speeds $\omega(t), \ov{\omega}(t)\in\mathbb{R}$ by the formal differential relations
\begin{equation}\label{diffrelone}
\frac{dx}{dt}=v \qquad \text{and} \qquad \frac{d\ov{x}}{dt}=\ov{v},
\end{equation}
together with
\begin{equation}\label{diffreltwo}
\frac{d\vartheta}{dt}=\omega \qquad \text{and} \qquad \frac{d\ov{\vartheta}}{dt}=\ov{\omega}.
\end{equation}
We gather the spatial and velocity data into single phase vectors $z$ and $\ov{z}$ given by
\begin{equation*}
z(t)=[x(t), \vartheta(t), v(t), \omega(t)]\in\mathcal{M}:=\mathbb{R}^{2}\times\mathbb{S}^{1}\times\mathbb{R}^{2}\times\mathbb{R},
\end{equation*}
and also
\begin{equation*}
\ov{z}(t)=[\ov{x}(t), \ov{\vartheta}(t), \ov{v}(t), \ov{\omega}(t)]\in\mathcal{M}:=\mathbb{R}^{2}\times\mathbb{S}^{1}\times\mathbb{R}^{2}\times\mathbb{R}.
\end{equation*}
We define the single phase vector which characterises the state of the whole system at time $t\in\mathbb{R}$ by $Z(t):=[z(t), \ov{z}(t)]\in\mathcal{M}^{2}$. As we do not wish that $\mathsf{P}(t)\cap\ov{\mathsf{P}}(t)$ have positive 2-dimensional Lebesgue measure for any time $t$, we stipulate that the range of the maps $t\mapsto Z(t)$ belong to the phase space $\mathcal{D}_{2}\equiv\mathcal{D}_{2}(\mathsf{P}_{\ast})$ defined by
\begin{equation*}
\mathcal{D}_{2}(\mathsf{P}_{\ast}):=\left\{Z\in\mathcal{M}^{2}\,:\,\mathrm{card}\,\mathsf{P}(x, \vartheta)\cap \mathsf{P}(\ov{x}, \ov{\vartheta})\leq 1\right\},
\end{equation*}
where $Z=[z, \ov{z}]$, with $z=[x, \vartheta, v, \omega]$ and $\ov{z}=[\ov{x}, \ov{\vartheta}, \ov{v}, \ov{\omega}]$. As it will be useful in what follows, we define the associated spatial projection operator $\Pi_{1}:\mathcal{D}_{2}\rightarrow\mathbb{R}^{4}\times\mathbb{T}^{2}$ by the rule
\begin{equation*}
\Pi_{1}Z:=[x, \ov{x}, \vartheta, \ov{\vartheta}] \quad \text{when}\hspace{2mm} Z=[z, \ov{z}]\in \mathcal{D}_{2}.
\end{equation*}
We also define the velocity projection operator $\Pi_{2}:\mathcal{D}_{2}\rightarrow \mathbb{R}^{6}$ by the rule
\begin{equation*}
\Pi_{2}Z:=[v, \ov{v}, \omega, \ov{\omega}] \quad \text{when}\hspace{2mm} Z=[z, \ov{z}]\in \mathcal{D}_{2}.
\end{equation*}
In order to be completely correct, we note that the differential relations \eqref{diffrelone} and \eqref{diffreltwo} only hold in general at those times $t\in\mathbb{R}$ for which $\mathsf{P}(t)\cap\mathsf{P}(t)=\varnothing$, i.e. the two-sided derivative limits in \eqref{diffrelone} and \eqref{diffreltwo} hold at those times $t$ when $\mathsf{P}(t)$ and $\ov{\mathsf{P}}(t)$ are not in collision with one another. At this point, it will prove helpful to make the following definition.
\begin{defn}
For any $Z_{0}\in\mathcal{D}_{2}$ and a map $Z:\mathbb{R}\rightarrow\mathcal{D}_{2}$ satisfying $Z(0)=Z_{0}$, we define the associated set of {\em collision times} $\mathcal{T}(Z_{0})\subseteq\mathbb{R}$ to be 
\begin{equation*}
\mathcal{T}(Z_{0}):=\left\{t\in\mathbb{R}\,:\,\mathrm{card}\,\mathsf{P}(t)\cap \ov{\mathsf{P}}(t)=1\right\}.
\end{equation*}
\end{defn}
In order to derive the equations of motion which govern the particles $\mathsf{P}(t)$ and $\ov{\mathsf{P}}(t)$, we first of all consider a class of 1-parameter families of operators $\{T_{t}\}_{t\in\mathbb{R}}$ ($T_{t}:\mathcal{D}_{2}\rightarrow\mathcal{D}_{2}$ for each $t\in\mathbb{R}$) for which the maps $t\mapsto \Pi_{1}T_{t}Z_{0}$ and $t\mapsto\Pi_{2}T_{t}Z_{0}$ have `reasonable' analytical properties. Indeed, in order to make concrete the primary objects of interest in this article, we make the following important definition. 
\begin{defn}\label{hpf}
We shall call a family of operators $\{T_{t}\}_{t\in\mathbb{R}}$ with $T_{t}:\mathcal{D}_{2}\rightarrow\mathcal{D}_{2}$ for each $t\in\mathbb{R}$ a \textbf{hard particle flow} on $\mathcal{D}_{2}$ if and only if for any $Z_{0}\in\mathcal{D}_{2}$, the map $t\mapsto \Pi_{1}T_{t}Z_{0}$ continuous and both left- and right-differentiable on $\mathbb{R}$ and the map $t\mapsto \Pi_{2}T_{t}Z_{0}$ is lower semi-continuous and left-differentiable on $\mathbb{R}$. Moreover, we stipulate that both $t\mapsto \Pi_{1}T_{t}Z_{0}$ and $t\mapsto\Pi_{2}T_{t}Z_{0}$ be differentiable at all times $t$ for which $T_{t}Z_{0}\in\mathcal{D}_{2}\setminus\partial\mathcal{D}_{2}$.
\end{defn}
The class of hard particle flows on $\mathcal{D}_{2}$ is evidently a rather large one. A basic question in classical mechanics is the following:  ``Which hard particle flows on $\mathcal{D}_{2}$ can one consider to be {\em physical}?'' To answer this question, and to specify in precise mathematical terms what we mean by {\em physical}, we appeal to Euler's Laws of Motion. When deriving an appropriate set of ODEs that govern the evolution of the phase map $t\mapsto Z(t)$, we divide our considerations into two cases, namely those times during which the dynamics is collision free, and those times at which a collision takes place. 
\subsection{Deriving the Equations of Motion when $\mathsf{P}(t)\cap\ov{\mathsf{P}}(t)=\varnothing$}\label{senseone}
Suppose a hard particle flow $\{T_{t}\}_{t\in\mathbb{R}}$ on the phase space $\mathcal{D}_{2}$ has been given. This flow gives rise naturally to a map $U:\mathbb{R}^{2}\times\mathbb{R}\times\mathcal{D}_{2}\rightarrow\mathbb{R}^{2}$ which provides the instantaneous linear velocity of any material point $x$ in $\mathbb{R}^{2}$ at any time $t$, once an initial condition $Z_{0}\in\mathcal{D}_{2}$ has been provided. Indeed, recall that if the centre of mass $x(t)$ of a planar rigid body $\mathsf{P}(t)$ translates with linear velocity $v(t)$, and $\mathsf{P}(t)$ rotates with angular speed $\omega(t)$, then the linear velocity of any other point on the body is expressed by the formula
\begin{equation*}
v(y, t)=v(t)+\omega(t)(y-x(t))^{\perp} \quad \text{for}\hspace{2mm}y\in R(\vartheta(t))\mathsf{P}_{\ast}+x(t),
\end{equation*}
where $y^{\perp}:=(-y_{2}, y_{1})$ for any given $y=(y_{1}, y_{2})\in\mathbb{R}^{2}$. As such, the map $U$ is given explicitly in terms of $\{T_{t}\}_{t\in\mathbb{R}}$ by
\begin{displaymath}
U(x, t; Z_{0})=\left\{
\begin{array}{ll}
v(t)+\omega(t)(x-x(t))^{\perp} & \quad \text{if}\hspace{2mm}x\in\mathsf{P}(t), \vspace{2mm}\\
\ov{v}(t)+\ov{\omega}(t)(x-\ov{x}(t))^{\perp} & \quad \text{if}\hspace{2mm}x\in\ov{\mathsf{P}}(t), \vspace{2mm}\\
0 & \quad \text{otherwise},
\end{array}
\right.
\end{displaymath}
where $\Pi_{1}T_{t}Z_{0}=[x(t), \ov{x}(t), \vartheta(t), \ov{\vartheta}(t)]$ and $\Pi_{2}T_{t}Z_{0}=[v(t), \ov{v}(t), \omega(t), \ov{\omega}(t)]$. Since $\{T_{t}\}_{t\in\mathbb{R}}$ is a hard particle flow (definition \ref{hpf}), it follows that $t\mapsto U(x, t, Z_{0})$ is a differentiable function at all $t$ for which $T_{t}Z_{0}\in\mathcal{D}_{2}\setminus\partial\mathcal{D}_{2}$.

We appeal to Euler's laws of motion in order to partition the class of hard particle flows into `physical' and `unphysical' flows. We henceforth assume that the motion of the hard particles $\mathsf{P}$ and $\ov{\mathsf{P}}$ takes place in the absence of external forces. Consider any $Z_{0}\in\mathcal{D}_{2}$ for which $\mathcal{T}(Z_{0})\neq \mathbb{R}$, and let us restrict our attention to the open set $I(Z_{0})\subseteq\mathbb{R}$ on which $\mathsf{P}(t)\cap\ov{\mathsf{P}}(t)=\varnothing$, i.e. where the map $t\mapsto \Pi_{2}T_{t}Z_{0}$ is differentiable. We now consider {\bf Euler's First Law of Motion} (\textsc{Truesdell} \cite{MR1162744}), which states that for any smooth evolution of smooth subsets $t\mapsto\Omega(t)\subseteq\mathbb{R}^{2}$, a physical hard particle flow should satisfy
\begin{equation}\label{efl}
\frac{d}{dt}\int_{\Omega(t)}U(x, t; Z_{0})\,dx=0 \quad \text{for} \hspace{2mm} t\in I(Z_{0}).
\end{equation}
Since we are free to choose the family of testing sets $\{\Omega(t)\,:\,t\in I(Z_{0})\}$ as we wish, we first pick it to be a family of smooth open sets such that $\mathsf{P}(t)\subset\Omega(t)$ together with $\Omega(t)\cap \ov{\mathsf{P}}(t)=\varnothing$ for all $t\in I(Z_{0})$. Similarly, we can also choose $\Omega(t)$ to contain particle $\ov{\mathsf{P}}(t)$ alone. As $t\mapsto T_{t}Z_{0}$ is differentiable on $I(Z_{0})$, identity \eqref{efl} reduces under these two choices to the ODEs
\begin{equation}\label{firstodes}
m\frac{dv}{dt}=0\quad \text{and}\quad m\frac{d\ov{v}}{dt}=0,
\end{equation}
where $m=\int_{\mathsf{P}_{\ast}}\,dy$ is the mass of the reference particle $\mathsf{P}_{\ast}$. This implies in particular that the total linear momentum of the initial datum $Z_{0}$ is conserved on $I(Z_{0})$. Thus, in the absence of external forces and collisions, Euler's first law simply reduces to the {\em conservation of linear momentum}.

It is now we turn to {\bf Euler's Second Law of Motion} (\textsc{Truesdell} \cite{MR1162744}), which states that
\begin{equation*}
\frac{d}{dt}\int_{\Omega(t)}(x-a)^{\perp}\cdot U(x, t; Z_{0})\,dx=0.
\end{equation*}
By appropriate choices of $\Omega(t)$, we discover that Euler's second law of motion reduces to
\begin{equation*}
\frac{d}{dt}\big(-m(a-x(t))^{\perp}\cdot v(t)+J\omega(t)\big)=\frac{d}{dt}\big(-m(a-\ov{x}(t))^{\perp}\cdot\ov{v}(t)+J\ov{\omega}(t)\big)=0,
\end{equation*}
where $J:=\int_{\mathsf{P}_{\ast}}|y|^{2}\,dy$ is the moment of inertia of the reference particle $\mathsf{P}_{\ast}$. By appealing to the ODEs \eqref{firstodes} derived above, we may infer that
\begin{equation*}
\frac{d\omega}{dt}=0\quad \text{and}\quad \frac{d\ov{\omega}}{dt}=0.
\end{equation*}
Therefore, it is clear that Euler's first and second laws together imply the conservation of linear and angular momentum for $T_{t}Z_{0}$ on $I(Z_{0})$. Importantly, one may check that Euler's first and second law imply that total kinetic energy is conserved in time, in the sense that
\begin{equation*}
\frac{d}{dt}\int_{\mathbb{R}^{2}}| U(x, t; Z_{0})|^{2}\,dx=0 \quad \text{for all}\hspace{2mm} t\in I(Z_{0}).
\end{equation*}
With this discussion in place, we now specify in precise terms what we mean by a {\em classical solution} to the ODEs derived from Euler's laws.


\subsection{A Dynamical System and its Boundary Conditions}\label{dynam}
Due to the possibility of particle collisions, we cannot expect the velocity maps $t\mapsto\Pi_{2}T_{t}Z_{0}$ to be differentiable both on the left and on the right on $\mathbb{R}$. As such, we separate out the information contained in Euler's ODEs into its left- and right-limits. We consider the following class of dynamical system, namely the evolution of two identical compact, strictly-convex sets $\mathsf{P}$ and $\ov{\mathsf{P}}$ (which are translations and rotations of the reference particle $\mathsf{P}_{\ast}$), whose phase trajectory $t\mapsto Z(t)\in\mathcal{D}_{2}$ satisfies the system of one-sided ODEs
\begin{displaymath}
\frac{d}{dt_{-}}\left[
\begin{array}{c}
x \\ \vartheta \\ v \\ \omega
\end{array}\right]=\left[
\begin{array}{c}
v_{-} \\ \omega_{-} \\ 0 \\ 0
\end{array}\right] \quad \text{and} \quad \frac{d}{dt_{-}}\left[
\begin{array}{c}
\ov{x} \\ \ov{\vartheta} \\ \ov{v} \\ \ov{\omega} 
\end{array}
\right]=\left[
\begin{array}{c}
\ov{v}_{-} \\ \ov{\omega}_{-} \\ 0 \\ 0
\end{array}
\right],
\end{displaymath}
in the classical sense for all $t\in\mathbb{R}$, where
\begin{equation*}
v_{-}(t):=\lim_{h\rightarrow 0-}\frac{x(t+h)-x(t)}{h} \qquad \text{and} \qquad \omega_{-}(t):=\lim_{h\rightarrow 0-}\frac{\vartheta(t+h)-\vartheta(t)}{h},
\end{equation*}
and similarly for the barred variables $\ov{v}_{-}$ and $\ov{\omega}_{-}$. We also ask that $t\mapsto Z(t)$ satisfies the system
\begin{displaymath}
\frac{d}{dt_{+}}\left[
\begin{array}{c}
x \\ \vartheta \\ v \\ \omega
\end{array}\right]=\left[
\begin{array}{c}
v_{+} \\ \omega_{+} \\ 0 \\ 0
\end{array}\right] \quad \text{and} \quad \frac{d}{dt_{+}}\left[
\begin{array}{c}
\ov{x} \\ \ov{\vartheta} \\ \ov{v} \\ \ov{\omega} 
\end{array}
\right]=\left[
\begin{array}{c}
\ov{v}_{+} \\ \ov{\omega}_{+} \\ 0 \\ 0
\end{array}
\right],
\end{displaymath}
in the classical sense for $t\in\mathbb{R}\setminus\mathcal{T}(Z_{0})$, where 
\begin{equation*}
v_{+}(t):=\lim_{h\rightarrow 0+}\frac{x(t+h)-x(t)}{h} \qquad \text{and} \qquad \omega_{+}(t):=\lim_{h\rightarrow 0+}\frac{\vartheta(t+h)-\vartheta(t)}{h},
\end{equation*}
and similarly for the barred variables. With this in place, we make the following definition.
\begin{defn}\label{classicalsolutions}
For a given initial datum $Z_{0}\in\mathcal{D}_{2}$, we say that $Z:\mathbb{R}\rightarrow\mathcal{D}_{2}$ is a (global-in-time) {\bf classical solution} of the above system of Euler's equations of motion if and only if $x, \ov{x}, \vartheta, \ov{\vartheta}$ are continuous piecewise linear on $\mathbb{R}$, and $v, \ov{v}, \omega, \ov{\omega}$ are lower-semicontinuous piecewise constant. Moreover, these maps satisfy the above ODEs pointwise on $\mathbb{R}$ for the left-derivatives, and pointwise on $\mathbb{R}\setminus\mathcal{T}(Z_{0})$ for the right-derivatives. Finally, $Z(0)=Z_{0}$.
\end{defn}
Evidently, the system of ODEs above is not enough to determine a family of flow operators $\{T_{t}\}_{t\in\mathbb{R}}$ on $\mathcal{D}_{2}$ uniquely. Indeed, one must specify how to update the dynamics at all collision times $\tau\in\mathcal{T}(Z_{0})$, i.e. for all $\tau$ such that
\begin{equation*}
\mathrm{card}\,\mathsf{P}(x(\tau), \vartheta(\tau))\cap\mathsf{P}(\ov{x}(\tau), \ov{\vartheta}(\tau))= 1.
\end{equation*}
When two compact, strictly-convex nonspherical particles are in contact at a single point, their configuration can be characterised (with respect to the reference particle $\mathsf{P}_{\ast}$) by an element $\beta$ of the 3-torus $\mathbb{T}^{3}$. In order to be able to construct a flow on $\mathcal{D}_{2}$, one must in turn construct an associated family of velocity {\em scattering maps} $\{\sigma_{\beta}\}_{\beta\in\mathbb{T}^{3}}$ on $\mathbb{R}^{6}$, each member of which sends elements in a set of `pre-collisional' velocity vectors to elements in a set of `post-collisional' velocity vectors (see section \ref{prepost} below for the precise definition of these sets). 

Not only this, one would ideally wish the family of flow operators $\{T_{t}\}_{t\in\mathbb{R}}$ on $\mathcal{D}_{2}$ to conserve the total linear momentum, angular momentum and kinetic energy of any given initial datum; consequently, any scattering map $\sigma_{\beta}:\mathbb{R}^{6}\rightarrow\mathbb{R}^{6}$ should also have this property (consult section \ref{algcon} for a discussion of the conserved quantities of the dynamics). However, it is shown in \cite{wilk} that such a family of scattering maps on $\mathbb{R}^{6}$ does not exist.
Nevertheless, we study a class of scattering matrices which gives rise to a hard particle flow on $\mathcal{D}_{2}$ that conserves total linear momentum and kinetic energy of all initial data. Before we can construct any flow associated with the above ODEs on $\mathcal{D}_{2}$, we must first find a convenient way by which to parameterise collision configurations. This is the subject of section \ref{paramet} below.

\subsection{Scattering Maps on $\mathbb{R}^{6}$}\label{scatteringmaps}
Scattering maps are the fundamental objects with which we work in this article. In particular, they must be constructed if one is to employ the existence theory for rigid body mechanics due to \textsc{Ballard} (see, in particular, hypothesis H3 \cite{ballard} p.212). In order to construct scattering maps, we must first find a careful parameterisation of all possible two-particle collision configurations, and then in turn specify what one means by pre- and post-collisional velocity vectors.
\subsubsection{Parameterising Collision Configurations}\label{paramet}
We now parameterise the set of all $Z\in\mathcal{D}_{2}$, up to translation, such that $\mathrm{card}\,\mathsf{P}(x, \vartheta)\cap\mathsf{P}(\ov{x}, \ov{\vartheta})=1$. In this direction, we consider what we call a {\em reference collision configuration} which will allow us to parameterise a general collision configuration of two particles by an element of the 3-torus $\mathbb{T}^{3}$. By considering the plane $\mathbb{R}^{2}$ furnished with polar co-ordinates, we make the problem of describing collision configurations considerably simpler. Indeed, as previously indicated, it will be of some help to consider the centre of mass of the reference particle $\mathsf{P}_{\ast}$ as at the origin of $\mathbb{R}^{2}$, which the polar map
\begin{displaymath}
x(\rho, \psi)=\left\{\begin{array}{lr}
(\rho\cos\psi, \rho\sin\psi) & \text{when}\hspace{2mm} (\rho, \psi)\in (0, \infty)\times \mathbb{S}^{1}, \vspace{2mm}\\
(0, 0) & \text{otherwise},
\end{array}
\right.
\end{displaymath}
co-ordinatises. 
We shall use $\mathsf{P}_{\ast}$ to define reference collision maps which are functions of the polar angle $\psi\in \mathbb{S}^{1}$ and the orientation $\theta\in\mathbb{S}^{1}$ of the particle exterior to the reference particle $\mathsf{P}_{\ast}$, namely $n_{\theta}=n_{\theta}(\psi), N_{\theta}=N_{\theta}(\psi), p_{\theta}=p_{\theta}(\psi), q_{\theta}=q_{\theta}(\psi)$ and $d_{\theta}=d_{\theta}(\psi)$; see Figure 1 below for an illustration of these quantities. They constitute the essential spatial data used to construct post-collisional velocities in a collision between two particles.

We begin by making the following definition.
\begin{defn}
Let $\psi, \theta\in \mathbb{S}^{1}$ be given. The {\bf distance of closest approach} $d_{\theta}(\psi)$ of the centres of mass of $\mathsf{P}_{\ast}$ and $\mathsf{P}(\cdot, \vartheta)$ for the given elevation angle $\psi$ is defined to be 
\begin{equation*}
d_{\theta}(\psi):=\inf\bigg\{d>0\,:\,\mathrm{card}\,\mathsf{P}_{\ast}\cap\left(R(\theta)\mathsf{P}_{\ast}+de(\psi)\right)=0\bigg\},
\end{equation*}
where $e(\psi):=(\cos\psi, \sin\psi)$. 
\end{defn}
With this basic and important quantity defined, we make another important definition.
\begin{defn}\label{reference}
We say that two particles $\mathsf{P}_{1}, \mathsf{P}_{2}\subset\mathbb{R}^{2}$ are in a {\em reference collision configuration} whenever $\mathsf{P}_{i}=\mathsf{P}_{\ast}$ (for some $i\in\{1, 2\}$) and there exist $\theta\in \mathbb{S}^{1}$ and $\psi\in \mathbb{S}^{1}$ such that $\mathsf{P}_{j}=R(\theta)\mathsf{P}_{\ast}+d_{\theta}(\psi)e(\psi)$, for $j\neq i$.
\end{defn} 
The other basic collision configuration quantities are now straightforward to characterise. We define the {\bf collision vector} $p=p_{\theta}(\psi)$ to be the unique element of the set
\begin{equation*}
\mathsf{P}_{\ast}\cap\left(R(\theta)\mathsf{P}_{\ast}+d_{\theta}(\psi)e(\psi)\right),
\end{equation*}
and the {\bf conjugate collision vector} $q=q_{\theta}(\psi)$ by
\begin{equation*}
q_{\theta}(\psi):=p_{\theta}(\psi)-d_{\theta}(\psi)e(\psi).
\end{equation*}
Since $\partial\mathsf{P}_{\ast}=\mathsf{P}_{\ast}\setminus\mathrm{int}(\mathsf{P}_{\ast})$ is a closed $C^{\omega}$ curve in $\mathbb{R}^{2}$ and can therefore be described locally by a smooth polar map $h_{\mathsf{P}_{\ast}}$, one can speak of the (outward) {\bf contact normal} $n=n_{\theta}(\psi)$ to the point of collision $p=p_{\theta}(\psi)$, which is given by
\begin{equation*}
n_{\theta}(\psi):=\frac{h_{\mathsf{P}_{\ast}}'(\alpha_{\theta}(\psi))^{\perp}}{|h_{\mathsf{P}_{\ast}}'(\alpha_{\theta}(\psi))^{\perp}|},
\end{equation*}
where
\begin{equation*}
\alpha_{\theta}(\psi):=\arctan\left(\frac{p_{\theta}(\psi)_{2}}{p_{\theta}(\psi)_{1}}\right).
\end{equation*}
The {\bf exclusion normal} $N_{\theta}=N_{\theta}(\psi)$ is defined to be the (outward) unit normal to the closed $C^{\omega}$ curve $\mathcal{C}_{\theta}$ given by
\begin{equation*}
\mathcal{C}_{\theta}:=\left\{d_{\theta}(\psi)e(\psi)\,:\,\psi\in\mathbb{S}^{1}\right\}.
\end{equation*}
Notice that in the case of hard disks (when $\mathsf{P}_{\ast}=B(0, r)$ for some $r>0$), this curve is simply a circle of radius $2r$, whence $N_{\theta}$ coincides identically with $n_{\theta}$. These basic vectors are illustrated in Figure 1 below. 

\begin{figure}\label{refconfig}
\begin{tikzpicture}[scale=0.6, every node/.style={scale=0.5}]]
\draw[line width=1.5pt, gray!25!white, densely dotted]  (0,0) circle (10cm);
\draw[line width=1.5pt, gray!25!white, densely dotted]  (0,0) circle (9cm);
\draw[line width=1.5pt, gray!25!white, densely dotted]  (0,0) circle (8cm);
\draw[line width=1.5pt, gray!25!white, densely dotted]  (0,0) circle (7cm);
\draw[line width=1.5pt, gray!25!white, densely dotted]  (0,0) circle (6cm);
\draw[line width=1.5pt, gray!25!white, densely dotted]  (0,0) circle (5cm);
\draw[line width=1.5pt, gray!25!white, densely dotted]  (0,0) circle (4cm);
\draw[line width=1.5pt, gray!25!white, densely dotted]  (0,0) circle (3cm);
\draw[line width=1.5pt, gray!25!white, densely dotted]  (0,0) circle (2cm);
\draw[line width=1.5pt, gray!25!white, densely dotted]  (0,0) circle (1cm);
\draw[line width=1.5pt, gray!25!white, densely dotted] (-12, 0) -- (12, 0);
\draw[->, line width=1.5pt, gray!25!white] (0, 0) -- (12, 0);
\draw[line width=1.5pt, gray!25!white, densely dotted] (-9, -9) -- (9, 9);
\draw[line width=1.5pt, gray!25!white, densely dotted] (-9, 9) -- (9, -9);
\draw[line width=1.5pt, gray!25!white, densely dotted] (0, -12) -- (0, 12);

\draw[line width=1.5pt, rotate=135, shift={(4.1 cm, -4 cm)}, dashed, gray!40!white] (-8.1, 0) -- (5.3, 0);

\draw[line width= 1.5pt, dotted, gray!40!white] (0, 0) -- (-8, 0);

\draw[line width= 1.5pt, dotted, gray!40!white] (0, 5.7) -- (-8, 5.7);

\draw[->, line width= 1.5pt] (-8, 0) -- (-8, 5.7);
\draw[<-, line width= 1.5pt] (-8, 0) -- (-8, 5.7);

\draw[line width=1.5pt, pattern=north west lines, pattern color=gray!25!white] \boundellipse{0,0}{5.3}{2};
\draw[line width=1.5pt, rotate=135, shift={(4.1 cm, -4 cm)}] \boundellipse{0,0}{5.3}{2};
\draw (-0.2, -0.2) -- (0.2, 0.2);
\draw (-0.2, 0.2) -- (0.2, -0.2);

\draw[rotate=135, shift={(4.1 cm, -4 cm)}]  (-0.2, -0.2) -- (0.2, 0.2);
\draw[rotate=135, shift={(4.1 cm, -4 cm)}]  (-0.2, 0.2) -- (0.2, -0.2);


\draw[line width= 1.5pt, dashed] (0, 0) -- (2.55, 1.77);

\draw[line width= 1.5pt, dashed] (-0.05, 5.7) -- (2.55, 1.77);

\draw[->, line width=1.5pt,] (2.55, 1.77)--(4, 8);

\draw[line width= 1.5pt, dashed, gray!40!white] (0, 0) -- (-0.05, 11.5);

\node[scale=2pt] at (-9,3) {$\displaystyle d_{\theta}(\psi)$};

\node[scale=2pt] at (5.6, 5.6) {$\theta$};

\node[scale=2pt] at (7.7, 7.7) {$\psi$};

\node[scale=2pt] at (2, 0.4) {$p_{\theta}(\psi)$};

\node[scale=2pt] at (1.8, 4.5) {$q_{\theta}(\psi)$};

\node[scale=2pt] at (4, 8.5) {$n_{\theta}(\psi)$};

\node[scale=2pt] at (0.3, -2.7) {$\mathsf{P}_{\ast}$};

\node[scale=2pt] at (-8, 9) {$\mathsf{P}=R(\theta)\mathsf{P}_{\ast}+d_{\theta}(\psi)e(\psi)$};

\draw[>=to, ->, thick] (10.4, 0) arc (0:90.1:10.4cm);

\draw[>=to, ->, thick] (7.5, 0) arc (0:103:7.5cm);

\end{tikzpicture}
\caption{An example of a reference configuration for $\mathsf{P}_{\ast}$ and $\mathsf{P}=R(\theta)\mathsf{P}_{\ast}+d_{\theta}(\psi)e(\psi)$}
\end{figure}
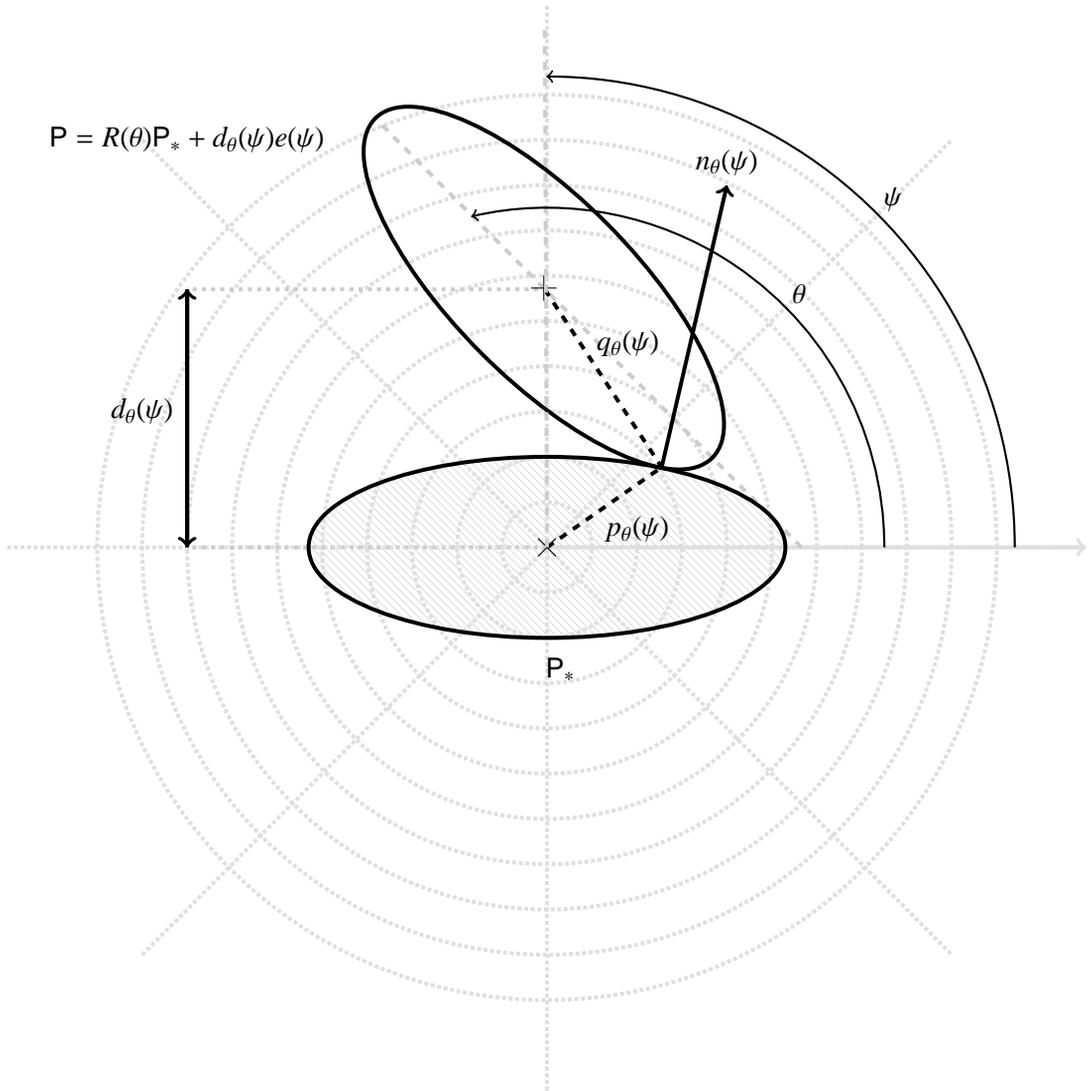
\subsubsection{General Collision Configurations}
When two particles $\mathsf{P}$ and $\ov{\mathsf{P}}$ in the dynamical system described above satisfy $\mathrm{card}\,\mathsf{P}(\tau)\cap\ov{\mathsf{P}}(\tau)=1$ for some $\tau\in\mathbb{R}$, we shall say they are in a {\em general collisional configuration}. Of course, it is not the case that they are necessarily in a {\em reference} collision configuration as defined above in definition \ref{reference}. In order to solve for the post-collisional linear velocities and angular speeds of two particles with arbitrary orientations (described by $\vartheta, \ov{\vartheta}\in \mathbb{S}^{1}$) and arbitrary relative position (described by $\psi\in\mathbb{S}^{1}$), it is expedient to relate general collisional configurations to the reference configuration introduced above. 

If $\mathsf{P}_{\ast}$ remains the standard reference particle, suppose $\mathsf{P}, \ov{\mathsf{P}}$ are of the form
\begin{equation*}
\mathsf{P}=R(\vartheta)\mathsf{P}_{\ast} \quad \text{and} \quad \ov{\mathsf{P}}=R(\ov{\vartheta})\mathsf{P}_{\ast}+x,
\end{equation*}
with $x\in\mathbb{R}^{2}$ such that $\mathrm{card}\,\mathsf{P}\cap\ov{\mathsf{P}}=1$, i.e. $\mathsf{P}$ and $\ov{\mathsf{P}}$ are in a collisional configuration. Thus, there exists an angle of elevation $\psi\in\mathbb{S}^{1}$ and a constant $\varrho=\varrho(\vartheta, \ov{\vartheta}, \psi)>0$ such that 
\begin{equation}\label{rotateme}
\mathsf{P}=R(\vartheta)\mathsf{P}_{\ast} \quad \text{and} \quad \ov{\mathsf{P}}=R(\ov{\vartheta})\mathsf{P}_{\ast}+\varrho(\vartheta, \ov{\vartheta}, \psi)e(\psi).
\end{equation}
In order to write down the appropriate distance of closest approach $d^{\ov{\vartheta}}_{\vartheta}$, together with the analogous collision vector $p^{\ov{\vartheta}}_{\vartheta}$, its conjugate $q^{\ov{\vartheta}}_{\vartheta}$ and the normals $n^{\ov{\vartheta}}_{\vartheta}$ and $N^{\ov{\vartheta}}_{\vartheta}$ in terms of the respective quantities $d_{\theta}, p_{\theta}, q_{\theta}, n_{\theta}$ and $N_{\theta}$ defined above, we perform some rotations. Acting on the system described in \eqref{rotateme} by the rotation matrix
\begin{displaymath}
R(\vartheta)^{T}=\left(
\begin{array}{cc}
\cos\vartheta & \sin\vartheta \\
-\sin\vartheta & \cos\vartheta
\end{array}
\right),
\end{displaymath}
we map $\mathsf{P}$ to $\mathsf{P}_{\ast}$ and $\ov{\mathsf{P}}$ to $R(\ov{\vartheta}-\vartheta)\mathsf{P}_{\ast}+\varrho(\vartheta, \ov{\vartheta}, \psi)e(\psi-\vartheta)$. This transformed system is now in a reference collision configuration. In particular, $\varrho(\vartheta, \ov{\vartheta}, \psi)=d_{\ov{\vartheta}-\vartheta}(\psi-\vartheta)$. Finally, by rotating back to the original configuration described by \eqref{rotateme}, it is clear that the basic collision quantities for two identical particles of orientations $\vartheta, \ov{\vartheta}\in \mathbb{S}^{1}$ whose centres of mass define a line of elevation $\psi$ with respect to the polar axis are the following:
\begin{align}
d_{\beta}=d^{\ov{\vartheta}}_{\vartheta}(\psi):=d_{\ov{\vartheta}-\vartheta}(\psi-\vartheta) \tag{distance between centres of mass} \\
p_{\beta}=p^{\ov{\vartheta}}_{\vartheta}(\psi):=R(\vartheta)p_{\ov{\vartheta}-\vartheta}(\psi-\vartheta) \tag{collision vector} \\
q_{\beta}=q^{\ov{\vartheta}}_{\vartheta}(\psi):=R(\vartheta)q_{\ov{\vartheta}-\vartheta}(\psi-\vartheta) \tag{conjugate collision vector} \\
N_{\beta}=N^{\ov{\vartheta}}_{\vartheta}(\psi):=R(\vartheta)N_{\ov{\vartheta}-\vartheta}(\psi-\vartheta) \tag{exclusion normal}
\end{align}
and
\begin{equation}\label{ndef}
n_{\beta}=n^{\ov{\vartheta}}_{\vartheta}(\psi):=R(\vartheta)n_{\ov{\vartheta}-\vartheta}(\psi-\vartheta). \tag{outward contact normal}
\end{equation}
These are illustrated in Figure 2 below. We work with these five fundamental vectors in all the sequel.
\begin{rem}
As we have done above, we shall often write the quantities such as $d^{\ov{\vartheta}}_{\vartheta}(\psi)$ simply as $d_{\beta}$ with $\beta=(\vartheta, \ov{\vartheta}, \psi)$ when the values of $\vartheta, \ov{\vartheta}, \psi\in \mathbb{S}^{1}$ are understood. It will often be convenient to use the notation $d^{\ov{\vartheta}}_{\vartheta}(\psi)$ whenever we emphasise that the parameters $(\vartheta, \ov{\vartheta})\in\mathbb{T}^{2}$ have been {\em fixed}, and $\psi\mapsto d^{\ov{\vartheta}}_{\vartheta}(\psi)$ is considered a function of $\psi$ alone. In this case, when the values of $(\vartheta, \ov{\vartheta})\in\mathbb{T}^{2}$ are understood, we shall simply write $d(\psi)$. This allows us to make the presentation of our arguments (especially those in section \ref{collnorm}) less cumbersome.
\end{rem}
\begin{figure}\label{generalcol}
\begin{tikzpicture}[scale=0.6, every node/.style={scale=0.5}]]
\draw[line width=1.5pt, gray!20!white, densely dotted]  (0,0) circle (10cm);
\draw[line width=1.5pt, gray!20!white, densely dotted]  (0,0) circle (9cm);
\draw[line width=1.5pt, gray!20!white, densely dotted]  (0,0) circle (8cm);
\draw[line width=1.5pt, gray!20!white, densely dotted]  (0,0) circle (7cm);
\draw[line width=1.5pt, gray!20!white, densely dotted]  (0,0) circle (6cm);
\draw[line width=1.5pt, gray!20!white, densely dotted]  (0,0) circle (5cm);
\draw[line width=1.5pt, gray!20!white, densely dotted]  (0,0) circle (4cm);
\draw[line width=1.5pt, gray!20!white, densely dotted]  (0,0) circle (3cm);
\draw[line width=1.5pt, gray!20!white, densely dotted]  (0,0) circle (2cm);
\draw[line width=1.5pt, gray!20!white, densely dotted]  (0,0) circle (1cm);
\draw[line width=1.5pt, gray!20!white, densely dotted] (-12, 0) -- (12, 0);
\draw[line width=1.5pt, gray!20!white, densely dotted] (-9, -9) -- (10, 10);
\draw[line width=1.5pt, gray!20!white, densely dotted] (-9, 9) -- (9, -9);
\draw[line width=1.5pt, gray!20!white, densely dotted] (0, -12) -- (0, 12);

\draw[<-, line width= 1pt, rotate=75, shift={(5.7cm, 9.9cm)}] (0, 0) -- (5, -3.26);
\draw[->, line width= 1pt, rotate=75, shift={(5.7cm, 9.9cm)}] (0, 0) -- (5, -3.26);

\draw[line width= 1.5pt, rotate=75, dashed, gray!40!white] (0, 0) -- (10, -6.52);
\draw[>=to, ->, line width=1pt, rotate=117, dashed, gray!40!white] (0, 0) -- (5.3, 0);

\draw[line width=1pt, rotate=135, dotted, gray!40!white] (0, 0) -- (11.3, 0);
\draw[line width=1pt, rotate=135, dotted, shift={(-0.3 cm, -5.95 cm)}, gray!40!white] (0, 0) -- (11.3, 0);

\draw[>=to, ->, line width=1pt, rotate=75, shift={(5 cm, -3.26 cm)}, dashed, gray!40!white] (-4.1, 0) -- (5.3, 0);
\draw[line width=1.5pt, rotate=117, pattern=north west lines, pattern color=gray!25!white] \boundellipse{0,0}{5.3}{2};
\draw[line width=1.5pt, rotate=75, shift={(5 cm, -3.26 cm)}] \boundellipse{0,0}{5.3}{2};
\draw (-0.2, -0.2) -- (0.2, 0.2);
\draw (-0.2, 0.2) -- (0.2, -0.2);

\draw[rotate=75, shift={(5 cm, -3.26 cm)}]  (-0.2, -0.2) -- (0.2, 0.2);
\draw[rotate=75, shift={(5 cm, -3.26 cm)}]  (-0.2, 0.2) -- (0.2, -0.2);


\node[scale=2pt] at (1.4, 3.1) {$\vartheta$};

\node[scale=2pt] at (8.3, 3.5) {$\ov{\vartheta}$};

\node[scale=2pt] at (10.7, 4.1) {$\psi$};

\node[scale=2pt] at (0.6, 0.7) {$p^{\ov{\vartheta}}_{\vartheta}(\psi)$};

\node[scale=2pt] at (4.7, 1.8) {$q^{\ov{\vartheta}}_{\vartheta}(\psi)$};

\node[scale=2pt] at (-6.9, 10.5) {$d^{\ov{\vartheta}}_{\vartheta}(\psi)$};


\draw[>=to, ->, thick] (11, 0) arc (0:42:11cm);

\draw[>=to, ->, thick] (8.5, 0) arc (0:52.5:8.5cm);

\draw[>=to, ->, thick] (3, 0) arc (0:116:3cm);

\draw[line width=1pt, rotate=-8, dashed] (0, 0) -- (2.35, 0);
\draw[line width=1pt, dashed] (4.43, 3.98) -- (2.37, -0.29);
\end{tikzpicture}
\caption{A general collision configuration}
\end{figure}
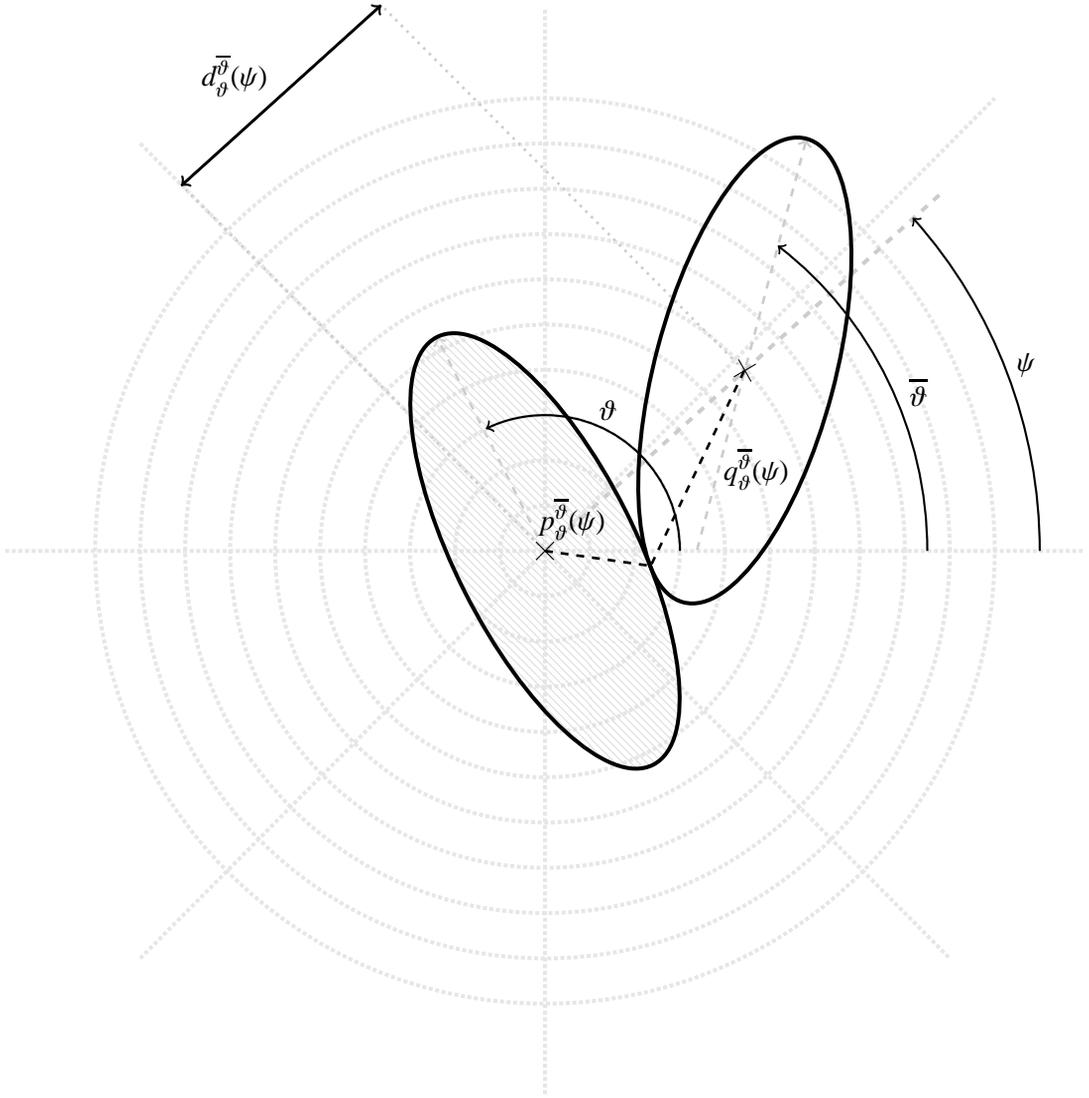
\subsubsection{Pre- and Post-collisional Velocities in $\mathbb{R}^{6}$}\label{prepost}
We now construct {\em scattering maps} on $\mathbb{R}^{6}$ which assign post-collisional velocities to pre-collisional velocities of two particles in a collision configuration in such a way that
\begin{equation*}
\mathrm{card}\,\mathsf{P}(x(t), \vartheta(t))\cap \mathsf{P}(\ov{x}(t), \ov{\vartheta}(t))\leq 1
\end{equation*}
for all $t$ in a sufficiently-small neighbourhood of a given collision time $\tau\in\mathcal{T}(Z_{0})$. Once we have such a map that uniquely updates the particle velocities, we may construct a global flow on phase space $\mathcal{D}_{2}$ corresponding to a classical solution of the system of governing ODEs introduced in section \ref{dynam} above using techniques from \cite{ballard}.

We now derive sets of {\em pre-} and {\em post}-collisional velocity vectors, and define what we mean by a scattering map. In order to do this, let us consider the auxiliary map $F:\mathbb{R}^{4}\times\mathbb{T}^{2}\rightarrow\mathbb{R}$ given by
\begin{equation*}
F(x, \ov{x}, \vartheta, \ov{\vartheta}):=|x-\ov{x}|-d^{\ov{\vartheta}}_{\vartheta}\left(\arctan\left[\frac{x_{2}-\ov{x}_{2}}{x_{1}-\ov{x}_{1}}\right]\right).
\end{equation*}
Clearly, $F(x, \ov{x}, \vartheta, \ov{\vartheta})>0$ if and only if $\mathsf{P}(x, \vartheta)\cap\mathsf{P}(\ov{x}, \ov{\vartheta})=\varnothing$; moreover, $F(x, \ov{x}, \vartheta, \ov{\vartheta})=0$ if and only if $\mathrm{card}\,\mathsf{P}(x, \vartheta)\cap\mathsf{P}(\ov{x}, \ov{\vartheta})=1$. We now introduce a hard particle dynamics $\{T_{t}\}_{t\in\mathbb{R}}$ associated with the ODE system in section \ref{dynam} above. Consider the maps $(x, \vartheta):\mathbb{R}\rightarrow \mathbb{R}^{2}\times \mathbb{S}^{1}$ and $(\ov{x}, \ov{\vartheta}):\mathbb{R}\rightarrow \mathbb{R}^{2}\times \mathbb{S}^{1}$ (with $\Pi_{1}T_{t}Z_{0}=[x, \ov{x}, \vartheta, \ov{\vartheta}]$) which satisfy
\begin{equation*}
\mathrm{card}\,\mathsf{P}(x(t), \vartheta(t))\cap \mathsf{P}(\ov{x}(t), \ov{\vartheta}(t))\leq 1
\end{equation*}
for all time $t\in\mathbb{R}$. We recall that, by assumption, $Z(t)=\Pi_{1}T_{t}Z_{0}$ is both left- and right-differentiable at all times, the only points at which right-derivatives do not necessarily agree with those on the left being the set of collision times $\mathcal{T}(Z_{0})$.

Consider now any collision time $\tau\in\mathcal{T}(Z_{0})$. Using the assumption of left-differentiability of the relevant phase maps, we have
\begin{equation*}
\frac{d}{dt_{-}}F(x(t), \ov{x}(t), \vartheta(t), \ov{\vartheta}(t))\bigg|_{t=\tau}\leq 0
\end{equation*}
for arbitrary $\tau\in\mathcal{T}(Z_{0})$, which a calculation reveals to be
\begin{align*}
\left(e(\psi)-\frac{1}{d^{\ov{\vartheta}}_{\vartheta}(\psi)}\frac{\partial d_{\ov{\vartheta}-\vartheta}}{\partial\psi}(\psi-\vartheta)e(\psi)^{\perp}\right)\cdot v_{-}-\left(e(\psi)-\frac{1}{d^{\ov{\vartheta}}_{\vartheta}(\psi)}\frac{\partial d_{\ov{\vartheta}-\vartheta}}{\partial\psi}(\psi-\vartheta)e(\psi)^{\perp}\right)\cdot \ov{v}_{-} \vspace{2mm} \\
+\left(\frac{\partial d_{\ov{\vartheta}-\vartheta}}{\partial\psi}(\psi-\vartheta)+\frac{\partial d_{\ov{\vartheta}-\vartheta}}{\partial\theta}(\psi-\vartheta)\right)\omega_{-}-\frac{\partial d_{\ov{\vartheta}-\vartheta}}{\partial\theta}(\psi-\vartheta)\ov{\omega}_{-}\leq 0,
\end{align*}
with $x(\tau), \ov{x}(\tau), \vartheta(\tau)$ and $\ov{\vartheta}(\tau)$ simply denoted by $x, \ov{x}, \vartheta$ and $\ov{\vartheta}$, respectively. Moreover, since the particles are in a collision configuration, there exists $\psi\in \mathbb{S}^{1}$ satisfying the identity $x-\ov{x}=d^{\ov{\vartheta}}_{\vartheta}(\psi)e(\psi)$. Now, we notice that the curve of closest approach
\begin{equation*}
\mathcal{C}^{\ov{\vartheta}}_{\vartheta}:=\left\{d^{\ov{\vartheta}}_{\vartheta}(\psi)e(\psi)\,:\,\psi\in \mathbb{S}^{1}\right\}
\end{equation*}
has (non-normalised) normal vectors
\begin{equation}\label{Nvector}
\widetilde{N}^{\ov{\vartheta}}_{\vartheta}(\psi):=e(\psi)-\frac{1}{d^{\ov{\vartheta}}_{\vartheta}(\psi)}\frac{\partial d_{\ov{\vartheta}-\vartheta}}{\partial\psi}(\psi-\vartheta)e(\psi)^{\perp},
\end{equation}
whose normalisation we denote by $N^{\ov{\vartheta}}_{\vartheta}(\psi):=\widetilde{N}^{\ov{\vartheta}}_{\vartheta}(\psi)/|\widetilde{N}^{\ov{\vartheta}}_{\vartheta}(\psi)|$. Moreover, we make the observation that
\begin{equation*}
d^{\ov{\vartheta}}_{\vartheta}(\psi)e(\psi)^{\perp}\cdot \widetilde{N}^{\ov{\vartheta}}_{\vartheta}(\psi)=-\frac{\partial d_{\ov{\vartheta}-\vartheta}}{\partial\psi}(\psi-\vartheta).
\end{equation*}
We therefore write the above inequality in the more compact form
\begin{align*}
N^{\ov{\vartheta}}_{\vartheta}(\psi)\cdot v_{-}-N^{\ov{\vartheta}}_{\vartheta}(\psi)\cdot \ov{v}_{-} \vspace{2mm}\\
+\left(r^{\ov{\vartheta}}_{\vartheta}(\psi)^{\perp}\cdot N^{\ov{\vartheta}}_{\vartheta}(\psi)-d^{\ov{\vartheta}}_{\vartheta}(\psi)e(\psi)^{\perp}\cdot N^{\ov{\vartheta}}_{\vartheta}(\psi)\right)\omega_{-}-r^{\ov{\vartheta}}_{\vartheta}(\psi)^{\perp}\cdot N^{\ov{\vartheta}}_{\vartheta}(\psi)\ov{\omega}_{-}\leq 0,
\end{align*}
where $r^{\ov{\vartheta}}_{\vartheta}(\psi)$ is the vector 
\begin{equation*}
r^{\ov{\vartheta}}_{\vartheta}(\psi):=-\frac{\partial d_{\ov{\vartheta}-\vartheta}}{\partial\theta}(\psi-\vartheta)e(\psi)^{\perp}.
\end{equation*}
As it is one of the most important quantities in all that follows, we make the following definition.
\begin{defn}\label{collnorm}
For any $\beta\in\mathbb{T}^{3}$, the {\bf collision normal} $\gamma_{\beta}\in\mathbb{R}^{6}$ is defined to be
\begin{displaymath}
\gamma_{\beta}:=\frac{1}{\sqrt{\Lambda_{\beta}}}\left[
\begin{array}{c}
N_{\beta}\\
-N_{\beta}\\
\left(r_{\beta}-d_{\beta}e(\psi)\right)^{\perp}\cdot N_{\beta}\\
-r_{\beta}^{\perp}\cdot N_{\beta}
\end{array}
\right],
\end{displaymath}
where 
\begin{equation}\label{lambdaish}
\Lambda_{\beta}:=\frac{2}{m}+\frac{1}{J}\left|\left(r_{\beta}-d_{\beta}e(\psi)\right)^{\perp}\cdot N_{\beta}\right|^{2}+\frac{1}{J}|r_{\beta}^{\perp}\cdot N_{\beta}|^{2}.
\end{equation}
\end{defn}
\begin{rem}
A quick calculation reveals that the collision normal $\gamma_{\beta}$ is not of unit norm. It will be useful rather often to employ the {\em unit collision normal} $\widehat{\gamma}_{\beta}:=M^{-1}\gamma_{\beta} $ in what follows.
\end{rem}
In the language of definition \ref{collnorm}, one then has that
\begin{equation*}
\frac{d}{dt_{-}}F(x(t), \ov{x}(t), \vartheta(t), \ov{\vartheta}(t))\bigg|_{t=\tau}\leq 0
\end{equation*}
if and only if
\begin{equation*}
\gamma_{\beta}\cdot V_{-}\leq 0,
\end{equation*}
where $V_{-}=[v_{-}, \ov{v}_{-}, \omega_{-}, \ov{\omega}_{-}]$. In a similar way, one can treat the post-collisional case and deduce that
\begin{equation*}
\frac{d}{dt_{+}}F(x(t), \ov{x}(t), \vartheta(t), \ov{\vartheta}(t))\bigg|_{t=\tau}\geq 0
\end{equation*}
if and only if
\begin{equation*}
\gamma_{\beta}\cdot V_{+}\geq 0.
\end{equation*}
Let a spatial configuration point $\beta\in\mathbb{T}^{3}$ be given and fixed. With the above discussion in mind, we define the set of {\em pre-collisional} velocities associated with the spatial configuration $\beta\in\mathbb{T}^{3}$ to be
\begin{equation*}
\Sigma_{\beta}^{-}:=\left\{V\in\mathbb{R}^{6}\,:\,V\cdot \gamma_{\beta} \leq 0\right\},
\end{equation*}
and the set of all {\em post-collisional} velocities to be
\begin{equation*}
\Sigma_{\beta}^{+}:=\left\{V\in\mathbb{R}^{6}\,:\, V\cdot \gamma_{\beta}\geq 0\right\}.
\end{equation*}
Evidently, $\mathbb{R}^{6}=\Sigma_{\beta}^{-}\cup\Sigma_{\beta}^{+}$. We denote the intersection $\Sigma_{\beta}^{-}\cap\Sigma_{\beta}^{+}$ of these two half-spaces by $\Sigma_{\beta}^{0}$. With these definitions in place, we can now say what we mean by a scattering map on $\mathbb{R}^{6}$.
\begin{defn}\label{scatmap}
We say that a bijective map $\sigma_{\beta}:\mathbb{R}^{6}\rightarrow\mathbb{R}^{6}$ is a {\bf scattering map} corresponding to the spatial configuration $\beta\in\mathbb{T}^{3}$ if and only if $\sigma_{\beta}(\Sigma_{\beta}^{-})=\Sigma_{\beta}^{+}$ and $\sigma_{\beta}\circ \sigma_{\beta}=\iota$ on $\mathbb{R}^{6}$.
\end{defn}
Suppose $\beta\in\mathbb{T}^{3}$, i.e. let the orientations and centres of mass of two particles in a collision configuration be given, and let $\sigma_{\beta}$ be an associated scattering map. By definition, 
\begin{equation}\label{pretopost}
V\cdot\gamma_{\beta}\leq 0\quad\Longrightarrow\quad \sigma_{\beta}[V]\cdot\gamma_{\beta}\geq 0,
\end{equation}
and also 
\begin{equation}\label{posttopre}
V\cdot\gamma_{\beta}\geq 0\quad\Longrightarrow\quad \sigma_{\beta}[V]\cdot\gamma_{\beta}\leq 0.
\end{equation}
It will be convenient in the rest of this article to write the above inequalities in what we shall call {\em quasi-momentum variables}. Consider the mass-inertia matrix $M\in\mathbb{R}^{6\times 6}$ given by
\begin{equation*}
M:=\mathrm{diag}(\sqrt{m}, \sqrt{m}, \sqrt{m}, \sqrt{m}, \sqrt{J}, \sqrt{J}).
\end{equation*}
Writing $P:=MV$ for a given $V\in\mathbb{R}^{6}$, and recalling that $\widehat{\gamma}_{\beta}=M^{-1}\gamma_{\beta}$, we can recast the above conditions as
\begin{equation}\label{Rpretopost}
P\cdot\widehat{\gamma}_{\beta}\leq 0 \quad \Longrightarrow\quad \rho_{\beta}[P]\cdot \widehat{\gamma}_{\beta}\geq 0,
\end{equation}
and
\begin{equation}\label{Rposttopre}
P\cdot\widehat{\gamma}_{\beta}\geq 0 \quad \Longrightarrow\quad \rho_{\beta}[P]\cdot \widehat{\gamma}_{\beta}\leq 0,
\end{equation}
where the transformed scattering map $\rho_{\beta}$ is given by
\begin{equation*}
\rho_{\beta}[P]:=M\sigma_{\beta}[M^{-1}P].
\end{equation*}
We write the associated transformed set of pre-collisional velocities as $\widehat{\Sigma}_{\beta}^{-}$, and the post-collisional velocities as $\widehat{\Sigma}_{\beta}^{+}$. 

There are many involutions $\sigma_{\beta}:\mathbb{R}^{6}\rightarrow\mathbb{R}^{6}$ which map the lower half-space $\Sigma_{\beta}^{-}$ to the upper half-space $\Sigma_{\beta}^{+}$. We now specify some conservation laws from classical mechanics, attributed to Euler's laws of motion, which should be respected by the hard particle flow $\{T_{t}\}_{t\in\mathbb{R}}$ on phase space $\mathcal{D}_{2}$. In particular, in view of the results in \cite{wilk}, we stipulate that the flow should conserve only total linear momentum and kinetic energy of given initial data $Z_{0}\in\mathcal{D}_{2}$.
\begin{figure}\label{generalcol}
\begin{tikzpicture}[scale=0.6, every node/.style={scale=0.5}]]
\draw[line width=1.5pt, gray!20!white, densely dotted]  (0,0) circle (10cm);
\draw[line width=1.5pt, gray!20!white, densely dotted]  (0,0) circle (9cm);
\draw[line width=1.5pt, gray!20!white, densely dotted]  (0,0) circle (8cm);
\draw[line width=1.5pt, gray!20!white, densely dotted]  (0,0) circle (7cm);
\draw[line width=1.5pt, gray!20!white, densely dotted]  (0,0) circle (6cm);
\draw[line width=1.5pt, gray!20!white, densely dotted]  (0,0) circle (5cm);
\draw[line width=1.5pt, gray!20!white, densely dotted]  (0,0) circle (4cm);
\draw[line width=1.5pt, gray!20!white, densely dotted]  (0,0) circle (3cm);
\draw[line width=1.5pt, gray!20!white, densely dotted]  (0,0) circle (2cm);
\draw[line width=1.5pt, gray!20!white, densely dotted]  (0,0) circle (1cm);
\draw[line width=1.5pt, gray!20!white, densely dotted] (-12, 0) -- (12, 0);
\draw[line width=1.5pt, gray!20!white, densely dotted] (-9, -9) -- (10, 10);
\draw[line width=1.5pt, gray!20!white, densely dotted] (-9, 9) -- (9, -9);
\draw[line width=1.5pt, gray!20!white, densely dotted] (0, -12) -- (0, 12);





\draw[line width=1.5pt, color=gray!25!white] \boundellipse{-6.3, 6.1}{5.3}{2};

\draw[line width=1.5pt, rotate=117, pattern=north west lines, pattern color=gray!25!white] \boundellipse{0,0}{5.3}{2};

\draw[->, line width=1.5pt] (-6.3, 6.1) -- (-9, 12.985);

\draw[->, line width=1.5pt] (-2.3, 4.75) -- (-5, 10);

\filldraw (-8.26, 2.7) circle (2pt);
\filldraw (8.26, -2.7) circle (2pt);
\filldraw (-8.23, 2.4) circle (2pt);
\filldraw (8.23, -2.4) circle (2pt);
\filldraw (-8.19, 2.1) circle (2pt);
\filldraw (8.19, -2.1) circle (2pt);
\filldraw (-8.15, 1.8) circle (2pt);
\filldraw (8.15, -1.8) circle (2pt);
\filldraw (-8.08, 1.5) circle (2pt);
\filldraw (8.08, -1.5) circle (2pt);
\filldraw (-8, 1.2) circle (2pt);
\filldraw (8, -1.2) circle (2pt);
\filldraw (-7.9, 0.9) circle (2pt);
\filldraw (7.9, -0.9) circle (2pt);
\filldraw (-7.79, 0.6) circle (2pt);
\filldraw (7.79, -0.6) circle (2pt);
\filldraw (-7.69, 0.3) circle (2pt);
\filldraw (7.69, -0.3) circle (2pt);

\filldraw (-7.95, 4.80) circle (2pt);
\filldraw (7.95, -4.80) circle (2pt);
\filldraw (-8.1, 4.5) circle (2pt);
\filldraw (8.1, -4.5) circle (2pt);
\filldraw (-8.2, 4.2) circle (2pt);
\filldraw (8.2, -4.2) circle (2pt);
\filldraw (-8.25, 3.9) circle (2pt);
\filldraw (8.25, -3.9) circle (2pt);
\filldraw (-8.31, 3.6) circle (2pt);
\filldraw (8.31, -3.6) circle (2pt);
\filldraw (-8.3, 3.3) circle (2pt);
\filldraw (8.3, -3.3) circle (2pt);
\filldraw (-8.29, 3) circle (2pt);
\filldraw (8.29, -3) circle (2pt);

\filldraw (-6.3, 6.1) circle (2pt);
\filldraw (6.3, -6.1) circle (2pt);
\filldraw (-6.6, 5.95) circle (2pt);
\filldraw (6.6, -5.95) circle (2pt);
\filldraw (-6.9, 5.8) circle (2pt);
\filldraw (6.9, -5.8) circle (2pt);
\filldraw (-7.2, 5.61) circle (2pt);
\filldraw (7.2, -5.61) circle (2pt);
\filldraw (-7.4, 5.45) circle (2pt);
\filldraw (7.4, -5.45) circle (2pt);
\filldraw (-7.6, 5.25) circle (2pt);
\filldraw (7.6, -5.25) circle (2pt);
\filldraw (-7.8, 5.05) circle (2pt);
\filldraw (7.8, -5.05) circle (2pt);

\filldraw (1.05, 6.5) circle (2pt);
\filldraw (-1.05, -6.5) circle (2pt);
\filldraw (-0.5, 6.75) circle (2pt);
\filldraw (0.5, -6.75) circle (2pt);
\filldraw (-1.5, 6.8) circle (2pt);
\filldraw (1.5, -6.8) circle (2pt);
\filldraw (-3.5, 6.79) circle (2pt);
\filldraw (3.5, -6.79) circle (2pt);
\filldraw (-4.2, 6.7) circle (2pt);
\filldraw (4.2, -6.7) circle (2pt);
\filldraw (-4.6, 6.6) circle (2pt);
\filldraw (4.6, -6.6) circle (2pt);
\filldraw (-5, 6.5) circle (2pt);
\filldraw (5, -6.5) circle (2pt);
\filldraw (-5.5, 6.35) circle (2pt);
\filldraw (5.5, -6.35) circle (2pt);
\filldraw (-5.9, 6.25) circle (2pt);
\filldraw (5.9, -6.25) circle (2pt);

\filldraw (4.99, 4.25) circle (2pt);
\filldraw (-4.99, -4.25) circle (2pt);
\filldraw (4.76, 4.5) circle (2pt);
\filldraw (-4.76, -4.5) circle (2pt);
\filldraw (4.51, 4.75) circle (2pt);
\filldraw (-4.51, -4.75) circle (2pt);
\filldraw (4.23, 5) circle (2pt);
\filldraw (-4.23, -5) circle (2pt);
\filldraw (3.92, 5.25) circle (2pt);
\filldraw (-3.92, -5.25) circle (2pt);
\filldraw (3.56, 5.5) circle (2pt);
\filldraw (-3.56, -5.5) circle (2pt);
\filldraw (3.1, 5.75) circle (2pt);
\filldraw (-3.1, -5.75) circle (2pt);
\filldraw (2.6, 6) circle (2pt);
\filldraw (-2.6, -6) circle (2pt);
\filldraw (1.95, 6.25) circle (2pt);
\filldraw (-1.95, -6.25) circle (2pt);

\filldraw (6.47, 2.25) circle (2pt);
\filldraw (-6.47, -2.25) circle (2pt);
\filldraw (6.33, 2.5) circle (2pt);
\filldraw (-6.33, -2.5) circle (2pt);
\filldraw (6.15, 2.75) circle (2pt);
\filldraw (-6.15, -2.75) circle (2pt);
\filldraw (5.99, 3) circle (2pt);
\filldraw (-5.99, -3) circle (2pt);
\filldraw (5.81, 3.25) circle (2pt);
\filldraw (-5.81, -3.25) circle (2pt);
\filldraw (5.64, 3.5) circle (2pt);
\filldraw (-5.64, -3.5) circle (2pt);
\filldraw (5.43, 3.75) circle (2pt);
\filldraw (-5.43, -3.75) circle (2pt);
\filldraw (5.22, 4) circle (2pt);
\filldraw (-5.22, -4) circle (2pt);

\filldraw (7.57,0) circle (2pt);
\filldraw (-7.57,0) circle (2pt);
\filldraw (7.48, 0.25) circle (2pt);
\filldraw (-7.48, -0.25) circle (2pt);
\filldraw (7.38, 0.5) circle (2pt);
\filldraw (-7.38, -0.5) circle (2pt);
\filldraw (7.28, 0.75) circle (2pt);
\filldraw (-7.28, -0.75) circle (2pt);
\filldraw (7.15, 1) circle (2pt);
\filldraw (-7.15, -1) circle (2pt);
\filldraw (7.03, 1.25) circle (2pt);
\filldraw (-7.03, -1.25) circle (2pt);
\filldraw (6.9, 1.5) circle (2pt);
\filldraw (-6.9, -1.5) circle (2pt);
\filldraw (6.74, 1.75) circle (2pt);
\filldraw (-6.74, -1.75) circle (2pt);
\filldraw (6.62, 2) circle (2pt);
\filldraw (-6.62, -2) circle (2pt);

\draw (-0.2, -0.2) -- (0.2, 0.2);
\draw (-0.2, 0.2) -- (0.2, -0.2);



%
%
%
%
\node[scale=2pt] at (-2, 8) {$n^{\ov{\vartheta}}_{\vartheta}(\psi)$};
\node[scale=2pt] at (-9.5, 10.5) {$N^{\ov{\vartheta}}_{\vartheta}(\psi)$};

\node[scale=2pt] at (8, -6) {$\mathcal{C}^{\ov{\vartheta}}_{\vartheta}$};


%
%

\end{tikzpicture}
\caption{A locus of closest approach with the exclusion and contact normals}
\end{figure}
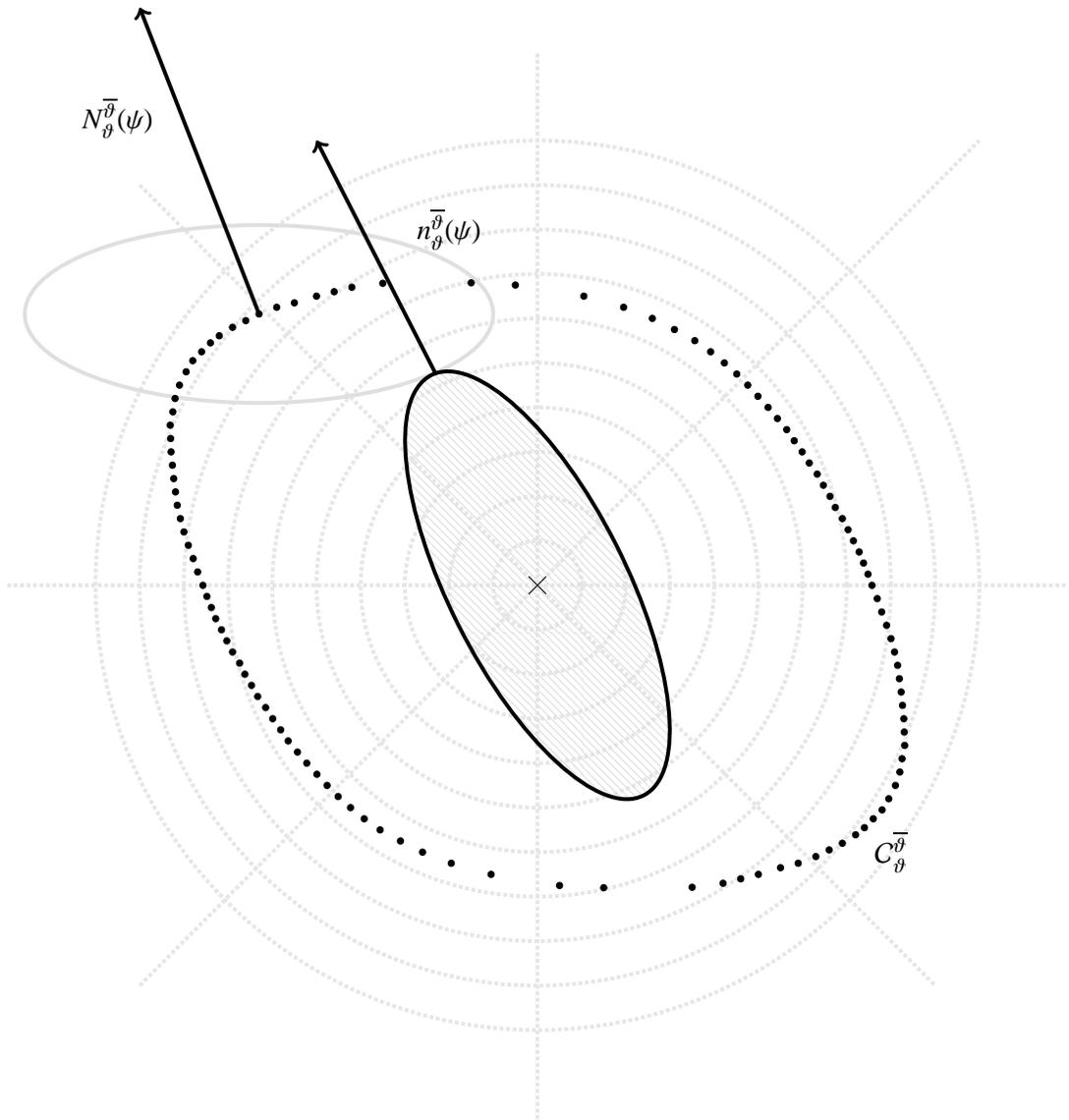
\subsection{Derivation of the Algebraic Constraints}\label{algcon}
Suppose the particles in collisional contact $\mathsf{P}:=R(\vartheta)\mathsf{P}_{\ast}$ and $\ov{\mathsf{P}}:=R(\ov{\vartheta})\mathsf{P}_{\ast}+d^{\ov{\vartheta}}_{\vartheta}(\psi)e(\psi)$ are given, together with their respective linear velocities and angular speeds $V\in\Sigma_{\beta}^{-}$, with $\beta=(\vartheta, \ov{\vartheta}, \psi)$. We seek post-collisional linear velocities and angular speeds $V'\in\Sigma_{\beta}^{+}$ such that there is conservation of total linear momentum and there is no loss of kinetic energy following collision. In what follows, unprimed quantities will denote pre-collisional ones, while those which are primed denote post-collisional ones.

Adhering to Euler's first law of motion, we stipulate that the values of the pre- and post-collisional velocities should satisfy the {\bf conservation of linear momentum}, i.e.
\begin{equation}\label{bawcolm}
\int_{\mathsf{P}(z(\tau))}v'_{\beta}(y, \tau)\,dy+\int_{\mathsf{P}(\ov{z}(\tau))}\ov{v}_{\beta}'(y, \tau)\,dy = \int_{\mathsf{P}(z(\tau))}v(y, \tau)\,dy+\int_{\mathsf{P}(\ov{z}(\tau))}\ov{v}(y, \tau)\,dy \tag{COLM},
\end{equation}
which since $v(y, t)=v(t)+\omega(t)(y-x(t))^{\perp}$ and $\ov{v}(y, t)=\ov{v}(t)+\ov{\omega}(t)(y-\ov{x}(t))^{\perp}$ (and similarly for the primed variables) reduces to
\begin{equation}\label{colm}
mv_{\beta}'+m\ov{v}_{\beta}'=mv+m\ov{v}.
\end{equation}
We also require that total kinetic energy be unchanged after the collision of the two particles. The {\bf conservation of kinetic energy} takes the form
\begin{equation}\label{bawcoke}
\int_{\mathsf{P}(z(\tau))}|v'_{\beta}(y, \tau)|^{2}\,dy+\int_{\mathsf{P}(\ov{z}(\tau))}|\ov{v}_{\beta}'(y, \tau)|^{2}\,dy =\int_{\mathsf{P}(z(\tau))}|v(y, \tau)|^{2}\,dy+\int_{\mathsf{P}(\ov{z}(\tau))}|\ov{v}(y, \tau)|^{2}\,dy, \tag{COKE}
\end{equation}
which reduces to 
\begin{equation}\label{coke}
m|v_{\beta}'|^{2}+J(\omega_{\beta}')^{2}+m|\ov{v}_{\beta}'|^{2}+J(\ov{\omega}_{\beta}')^{2}=m|v|^{2}+J\omega^{2}+m|\ov{v}|^{2}+J\ov{\omega}^{2}.
\end{equation}
Expressing the above conservation laws in scattering map notation, we find that \eqref{colm} takes the form
\begin{displaymath}
\left(
\begin{array}{c}
\sigma_{\beta}[V]_{1}+\sigma_{\beta}[V]_{3} \\
\sigma_{\beta}[V]_{2}+\sigma_{\beta}[V]_{4}
\end{array}
\right)=\left(
\begin{array}{c}
V_{1}+V_{3}\\
V_{2}+V_{4}
\end{array}
\right),
\end{displaymath}
while \eqref{coke} takes the form
\begin{equation*}
|M\sigma_{\beta}[V]^{2}=|MV|^{2},
\end{equation*}
where $V=[v, \ov{v}, \omega, \ov{\omega}]$. As claimed above, in order to prove \textsc{Theorem} \ref{poorflow} (or, rather, the more precise statement \ref{flow}), we must first construct a family of scattering maps $\{\sigma_{\beta}\}_{\beta\in\mathbb{T}^{3}}$ on $\mathbb{R}^{6}$, each member of which conserves total linear momentum and kinetic energy. This is the aim of the following section.

\subsection{Construction of a Dynamics for Euler's Equations on $\mathcal{D}_{2}$}\label{connydym}
We now aim to prove the following more precisely-stated form \textsc{Theorem} \ref{poorflow}.
\begin{thm}\label{flow}
Suppose $\mathsf{P}_{\ast}\subset\mathbb{R}^{2}$ is compact and strictly-convex with boundary $\partial \mathsf{P}_{\ast}$ of class $C^{\omega}$. For each $Z_{0}\in\mathcal{D}_{2}(\mathsf{P}_{\ast})$, there exists a global-in-time classical solution $Z(t)=T_{t}Z_{0}$ of Euler's equations with the property that
\begin{displaymath}
\left(
\begin{array}{c}
(\Pi_{2}T_{t}Z_{0})_{1}+(\Pi_{2}T_{t}Z_{0})_{3} \\
(\Pi_{2}T_{t}Z_{0})_{2}+(\Pi_{2}T_{t}Z_{0})_{4}
\end{array}
\right)=
\left(
\begin{array}{c}
(\Pi_{2}Z_{0})_{1}+(\Pi_{2}Z_{0})_{3} \\
(\Pi_{2}Z_{0})_{2}+(\Pi_{2}Z_{0})_{4}
\end{array}
\right)\quad \text{for all} \hspace{2mm} t\in\mathbb{R},
\end{displaymath}
and
\begin{equation*}
|M\Pi_{2}T_{t}Z_{0}|^{2}=|M\Pi_{2}Z_{0}|^{2}\quad \text{for all} \hspace{2mm} t\in\mathbb{R}.
\end{equation*}
\end{thm}
Notice that the above theorem makes no claim on uniqueness of solutions. They are, however, unique {\em with respect to a fixed family of scattering matrices} $\{\sigma_{\beta}\}_{\beta\in\mathbb{T}^{3}}$. In other words, once a family of scattering matrices has been chosen and fixed, the classical solutions of Euler's equations constructed using the theory of \cite{ballard} are unique. As such, we must make a choice regarding with which family of scattering maps we wish to work. Since the study of linear scattering maps and their corresponding collision invariants is made possible by means of group theoretic arguments for subgroups of the orthogonal group $\mathrm{O}(6)$ (see section \ref{mainsection} below), we subsequently focus on the case where scattering maps $\sigma_{\beta}:\mathbb{R}^{6}\rightarrow\mathbb{R}^{6}$ are matrices. One could construct solutions of the ODEs in the case when the scattering family $\{\sigma_{\beta}\}_{\beta\in\mathbb{T}^{3}}$ is a collection of {\em nonlinear} maps on $\mathbb{R}^{6}$. We do not, however, pursue this idea any further here. 
\subsubsection{The case of linear scattering $\sigma_{\beta}:\mathbb{R}^{6}\rightarrow\mathbb{R}^{6}$}
We establish the following preliminary result.
\begin{prop}\label{linearscat}
For a given $\beta\in\mathbb{T}^{3}$, let $\sigma_{\beta}$ be a linear scattering map which conserves kinetic energy and linear momentum, i.e. $\sigma_{\beta}[V]$ satisfies \eqref{colm} and \eqref{coke} and for all $V\in\mathbb{R}^{6}$. Then $\sigma_{\beta}$ is necessarily of the form
\begin{equation*}
\sigma_{\beta}=M^{-1}\left(\widehat{E}_{1}\otimes \widehat{E}_{1} + \widehat{E}_{2}\otimes\widehat{E}_{2}+\sum_{i=3}^{5}\lambda_{i}(\beta)\widehat{E}_{i}(\beta)\otimes \widehat{E}_{i}(\beta) -\widehat{\gamma}_{\beta}\otimes \widehat{\gamma}_{\beta}\right)M,
\end{equation*}
where $\widehat{E}_{1}=(\frac{1}{\sqrt{2}}, 0, \frac{1}{\sqrt{2}}, 0, 0, 0)$ $\widehat{E}_{2}=(0, \frac{1}{\sqrt{2}}, 0, \frac{1}{\sqrt{2}}, 0, 0)$, $\{\widehat{E}_{i}(\beta)\}_{i=3}^{5}$ is any orthonormal basis for $\mathrm{span}\{\widehat{E}_{1}, \widehat{E}_{2}, \widehat{\gamma}_{\beta}\}^{\perp}$, $\lambda_{i}(\beta)\in\{-1, 1\}$ and $\widehat{\gamma}_{\beta}$ is the unit collision normal \eqref{collnorm}.
\end{prop}
\begin{proof}
It will be convenient to consider the problem cast in quasi-momentum variables as introduced above in section \ref{prepost}. Indeed, given the scattering map $\sigma_{\beta}$ we define the map $\rho_{\beta}[P]:=M\sigma_{\beta}[M^{-1}P]$ for $P\in\mathbb{R}^{6}$. Since $\sigma_{\beta}$ is linear if and only if $\rho_{\beta}$ is linear, we may suppose that $\rho_{\beta}[P]=R_{\beta}P$ for some $R_{\beta}\in\mathrm{GL}(6)$. Moreover, we also infer that $\rho_{\beta}$ is an involution on $\mathbb{R}^{6}$, whence $R_{\beta}^{2}=I$. It will now prove useful to consider the spectral structure of $R_{\beta}$.

We first note that since the conservation of kinetic energy \eqref{coke} implies that $|R_{\beta}P|^{2}=|P|^{2}$ for all $P\in\mathbb{R}^{6}$, it follows that $R_{\beta}\in\mathrm{O}(6)$. Moreover, $R_{\beta}$ can only have real eigenvalues $\lambda$ with $|\lambda|=1$. Now, the conservation of linear momentum
\begin{displaymath}
\left(
\begin{array}{c}
\sigma_{\beta}[V]_{1}+\sigma_{\beta}[V]_{3} \\
\sigma_{\beta}[V]_{2}+\sigma_{\beta}[V]_{4}
\end{array}
\right)=\left(
\begin{array}{c}
V_{1}+V_{3} \\
V_{2} + V_{4}
\end{array}
\right)
\end{displaymath}
implies that 
\begin{equation*}
R_{\beta}P\cdot E_{1}=P\cdot E_{1} \quad \text{and}\quad R_{\beta}P\cdot E_{2}=P\cdot E_{2} \quad \text{for all}\hspace{2mm} P\in\mathbb{R}^{6},
\end{equation*}
where $E_{1}=(1, 0, 1, 0, 0, 0)$ and $E_{2}=(0, 1, 0, 1, 0, 0)$. We immediately infer that $E_{1}$ and $E_{2}$ are eigenvectors of $R_{\beta}$ both with eigenvalue 1, since $R_{\beta}^{T}=R_{\beta}$. Appealing to the fact that $R_{\beta}$ must satisfy the inequalities \eqref{Rpretopost} and \eqref{Rposttopre} above, since $\sigma_{\beta}$ was assumed to be a scattering map, we deduce that $R_{\beta}\widehat{\gamma}_{\beta}=-\widehat{\gamma}_{\beta}$, whence the unit collision normal $\widehat{\gamma}_{\beta}$ is another eigenvector of $R_{\beta}$ with eigenvalue $-1$. 

We restrict our attention to the subspace of $\mathbb{R}^{6}$ orthogonal to $\widehat{\gamma}_{\beta}$, namely $\widehat{\Sigma}_{\beta}^{0}:=\widehat{\Sigma}_{\beta}^{-}\cap\widehat{\Sigma}_{\beta}^{+}$. Setting 
\begin{equation}\label{unitvec}
\widehat{E}_{1}:=\left(\frac{1}{\sqrt{2}}, 0, \frac{1}{\sqrt{2}}, 0, 0, 0\right) \quad  \text{and} \quad \widehat{E}_{2}:=\left(0, \frac{1}{\sqrt{2}}, 0, \frac{1}{\sqrt{2}}, 0, 0\right),
\end{equation} 
one may check that $\widehat{E}_{1}\cdot \widehat{\gamma}_{\beta}=\widehat{E}_{2}\cdot \widehat{\gamma}_{\beta}=0$, while evidently $\widehat{E}_{1}\cdot \widehat{E}_{2}=0$. Let us consider any orthonormal basis of $\widehat{\Sigma}_{\beta}^{0}$ containing $\widehat{E}_{1}$ and $\widehat{E}_{2}$, namely $B_{\beta}:=\{\widehat{E}_{1}, \widehat{E}_{2}\}\cup\{\widehat{E}_{3}(\beta), \widehat{E}_{4}(\beta), \widehat{E}_{5}(\beta)\}$, where each $\widehat{E}_{i}(\beta)$ is allowed to depend on the spatial configuration $\beta\in\mathbb{T}^{3}$. One may then verify that {\em any} matrix of the form
\begin{equation*}
R_{\beta}:=\widehat{E}_{1}\otimes \widehat{E}_{1} + \widehat{E}_{2}\otimes\widehat{E}_{2}+\sum_{i=3}^{5}\lambda_{i}(\beta)\widehat{E}_{i}(\beta)\otimes \widehat{E}_{i}(\beta) -\widehat{\gamma}_{\beta}\otimes \widehat{\gamma}_{\beta}
\end{equation*}
with $\lambda_{i}(\beta)\in\{-1, 1\}$ is a bijective linear involution which maps $\widehat{\Sigma}_{\beta}^{-}$ to $\widehat{\Sigma}_{\beta}^{+}$. Moreover, transforming back from quasi-momentum variables, the associated scattering matrix $\sigma_{\beta}:=M^{-1}R_{\beta}M$ conserves the total linear momentum and kinetic energy of its argument. The proof of the proposition follows.
\end{proof}
Evidently, as we have such a large family of scattering matrices which conserve both linear momentum and kinetic energy, it is prudent to specify another natural condition on each matrix $\sigma_{\beta}$ to obtain a unique family of matrices $\{\sigma_{\beta}\}_{\beta\in\mathbb{T}^{3}}$ to which we can turn our attention. At this point, it is helpful to consider the case of hard disks. 
\subsubsection{Comparison with the Case of Hard Disks}\label{compare}
If we have developed a suitable extension of the classical scattering of hard disks to the more general compact, strictly-convex particle setting, the associated scattering matrix $\sigma_{\beta}$ should reduce essentially to the classical Boltzmann scattering matrix \eqref{boltzscat} when $\mathsf{P}_{\ast}$ is chosen to be a disk. We consider the case $\mathsf{P}_{\ast}=\mathsf{B}_{\ast}$ (the closed unit disk in $\mathbb{R}^{2}$). As the classical Boltzmann scattering matrices are unique in the class of all maps on $\mathbb{R}^{4}$ which conserve total linear momentum, angular momentum and kinetic energy of particles (and which also enforce non-penetration), we do not have mixing of pre-collisional linear velocities and angular speeds following collision. With this observation in mind, we consider the block scattering matrix defined on $\mathbb{R}^{6}$ by
$$
\left(
\begin{array}{cc}
\tempb  & 0_{2}\\ \hline
 0_{4} &\tempd
\end{array}\right)\in\mathbb{R}^{6\times 6} \quad \text{for}\hspace{2mm}\psi\in\mathbb{S}^{1},
$$
with $0_{m}, I_{m}\in\mathbb{R}^{m\times m}$ and $\widehat{\gamma}(\psi)=\frac{1}{\sqrt{2}}[e(\psi), -e(\psi)]$. Notably, this matrix is the identity map when restricted to the set $\Sigma_{\beta}^{0}$. Motivated by this observation, we have the following corollary to proposition \ref{linearscat} above. 
\begin{cor}
Suppose $\sigma_{\beta}$ is a scattering matrix satisfying the hypotheses of proposition \ref{linearscat} which is the identity map when restricted to $\Sigma_{\beta}^{0}=\Sigma_{\beta}^{-}\cap \Sigma_{\beta}^{+}$. Then $\sigma_{\beta}$ is necessarily of the form $\sigma_{\beta}=M^{-1}(I-2\widehat{\gamma}_{\beta}\otimes \widehat{\gamma}_{\beta})M$.
\end{cor} 
\begin{proof}
Let $B_{\beta}=\{\widehat{E}_{i}\}_{i=1}^{5}$ be any orthonormal basis for $\Sigma_{\beta}^{0}$ which contains the vectors $\widehat{E}_{1}$ and $\widehat{E}_{2}$ given above by \eqref{unitvec}. Since by assumption $\sigma_{\beta}|_{\Sigma_{\beta}^{0}}=\iota$, it follows that $\lambda_{i}(\beta)=1$ for $i=3, 4, 5$. Now, using the fact that
\begin{equation*}
I=\sum_{i=1}^{5}\widehat{E}_{i}\otimes \widehat{E}_{i}+\widehat{\gamma}_{\beta}\otimes \widehat{\gamma}_{\beta},
\end{equation*}
we find $R_{\beta}=I-2\widehat{\gamma}_{\beta}\otimes\widehat{\gamma}_{\beta}$. Transforming back to velocity variables $V$ from quasi-momentum variables $P$, we obtain $\sigma_{\beta}[V]=M^{-1}\left(I-2\widehat{\gamma}_{\beta}\otimes\widehat{\gamma}_{\beta}\right)M$, which yields the assertion of the corollary.
\end{proof}
As such, the derived family of scattering matrices reduces to the family of Boltzmann scattering matrices (which is the identity map when restricted to the factors of $\mathbb{R}^{6}$ describing angular speed) when the reference particle $\mathsf{P}_{\ast}$ is chosen to be a disk. With this concrete family of scattering matrices in hand, we now look to construct global-in-time classical solutions to Euler's equations on $\mathcal{D}_{2}$.
\subsubsection{Construction of Global-in-time Classical Solutions on $\mathcal{D}_{2}$}\label{localintime}
We now offer some brief comments that establish \textsc{Theorem} \ref{flow}, the proof of which follows swiftly from the construction of the scattering matrices $\sigma_{\beta}=M^{-1}(1-2\widehat{\gamma}_{\beta}\otimes\widehat{\gamma}_{\beta})M$ and an application of theorem 10 in \textsc{Ballard} \cite{ballard}. We do not discuss technical details of the proof here, and refer the reader to (\cite{ballard}, section 4) for details. Given that $\partial\mathsf{P}_{\ast}$ is of class $C^{\omega}$ and that there is no externally-imposed force in the equations of motion (S\textsuperscript{--}) and (S\textsuperscript{+}), it follows that for each initial datum $Z_{0}\in\mathcal{D}_{2}$\footnote{To be precise, if $Z_{0}$ is taken to lie in $\partial\mathcal{D}_{2}$ (namely the initial condition describes a collision configuration) then for consistency we should only allow for initial velocities $\Pi_{2}Z_{0}$ to lie in $\Sigma_{\beta}^{-}$, where $\beta\in\mathbb{T}^{3}$ is determined by $\Pi_{1}Z_{0}$.} there exists a unique piecewise linear map $t\mapsto [x(t), \ov{x}(t), \vartheta(t), \ov{\vartheta}(t)]$ with 
\begin{equation*}
[x(0), \ov{x}(0), \vartheta(0), \ov{\vartheta}(0)]=\Pi_{1}Z_{0} \quad \text{and}\quad \frac{d}{dt_{-}}[x(t), \ov{x}(t), \vartheta(t), \ov{\vartheta}(t)]\bigg|_{t=0}=\Pi_{2}Z_{0},
\end{equation*}
which satisfies (S\textsuperscript{--}) and (S\textsuperscript{+}) on $\mathbb{R}$ and $\mathbb{R}\setminus\mathcal{T}(Z_{0})$, respectively. Moreover, for every such initial datum $Z_{0}\in\mathcal{T}(Z_{0})$ the set of all collision times $\mathcal{T}(Z_{0})$ is finite, i.e. $\mathcal{T}(Z_{0})=\{\tau_{j}\}_{j=1}^{M}$ with $M=M(Z_{0})\in\mathbb{N}$, with the property that for each $t\in (\tau_{j}, \tau_{j+1}]$, there exists a left-neighbourhood of $t$ on which $t\mapsto [x(t), \ov{x}(t), \vartheta(t), \ov{\vartheta}(t)]$ is analytic. Importantly, uniqueness of classical solutions allows us to define a hard particle flow $\{T_{t}\}_{t\in\mathbb{R}}$ on $\mathcal{D}_{2}$ with the property that total linear momentum and kinetic energy of initial data is conserved for all time, and for which the colliding particles experience at most finitely-many collisions on bounded time intervals. 

It is also important to emphasise that in order to make use of the general existence theory in \cite{ballard}, a family of scattering maps must be provided as data for the problem. As such, classical solutions are only unique with respect to the given family of scattering maps under consideration. It would be possible to construct another distinct hard particle flow on $\mathcal{D}_{2}$ that conserves total linear momentum and kinetic energy if one constructs a family of {\em nonlinear} scattering maps $\{\sigma_{\beta}\}_{\beta\in\mathbb{T}^{3}}$ on $\mathbb{R}^{6}$ satisfying the same property. As intimated above, we do not address this problem in this article.
\subsubsection{An `Almost Physical' Family of Matrices}
It is important to record the fact here that the matrix $u_{\beta}:=M^{-1}\left(I-2\widehat{\eta}_{\beta}\otimes\widehat{\eta}_{\beta}\right)M\in\mathbb{R}^{6\times 6}$,
where the unit vector $\widehat{\eta}_{\beta}\in\mathbb{R}^{6}$ is given by
\begin{displaymath}
\widehat{\eta}_{\beta}:=\frac{1}{\sqrt{\frac{2}{m}+\frac{1}{J}|p_{\beta}^{\perp}\cdot n_{\beta}|^{2}+\frac{1}{J}|q_{\beta}^{\perp}\cdot n_{\beta}|^{2}}}M^{-1}\left[
\begin{array}{c}
n_{\beta} \\
-n_{\beta} \\
p_{\beta}^{\perp}\cdot n_{\beta} \\
-q_{\beta}^{\perp}\cdot n_{\beta}
\end{array}
\right]\in\mathbb{R}^{6},
\end{displaymath}
conserves the total linear momentum, angular momentum and kinetic energy of its argument, but it is {\em not} a scattering map in the sense of definition \ref{scatmap} above. In particular, one can find collision configurations $\beta_{\ast}\in\mathbb{T}^{3}$ and associated pre-collisional velocities $V_{\ast}\in\Sigma_{\beta_{\ast}}^{-}\setminus\Sigma_{\beta_{\ast}}^{0}$ that satisfy
\begin{equation*}
u_{\beta_{\ast}}V_{\ast}=V_{\ast}
\end{equation*}
and which therefore lead to interpenetration of the particles when the dynamics of (S\textsuperscript{--}) and (S\textsuperscript{+}) is continued after collision. As a result, it cannot be used to construct a hard particle flow on $\mathcal{D}_{2}$, but it can be used to construct a family of flow operators on $\mathcal{M}^{2}$ corresponding to Euler's equations presented in section \ref{dynam}.

We make the rather na\"{i}ve comment that in the Boltzmann-Grad limit of the BBGKY hierarchy, `particles become points' and so it makes no sense to speak of non-penetration of particles for the limiting system as the number of particles $N\rightarrow\infty$. As such, one could argue that the family of maps $\{u_{\beta}\}_{\beta\in\mathbb{T}^{3}}$ would nevertheless be suitable to establish a kinetic model for the average behaviour of rarified gases composed of compact, strictly-convex particles. Indeed, the main result \textsc{Theorem} \ref{awesome} on characterisation of collision invariants for non-spherical particles in this article also holds for the family $\{u_{\beta}\}_{\beta\in\mathbb{T}^{3}}$, under the weaker condition that $\partial\mathsf{P}_{\ast}$ be of class $C^{1}$, as opposed to analytic. 


\section{Collision Invariants for Compact, Strictly-convex Particles}\label{mainsection}
We now turn to the proof of the main result of this article. We firstly define the analogue of classical collision invariants in the case when the underlying particles are not disks.
\begin{defn}
Let $\mathcal{S}=\{\sigma_{\beta}\}_{\beta\in\mathbb{T}^{3}}$ be a family of maps on $\mathbb{R}^{6}$. A measurable function $\varphi:\mathbb{R}^{2}\times\mathbb{R}\times\mathbb{S}^{1}\rightarrow\mathbb{R}$ is said to be an {\em $\mathcal{S}$-collision invariant} if and only if it satisfies the functional equation
\begin{equation}\label{collinv}
\varphi(v_{\beta}', \omega_{\beta}', \vartheta)+\varphi(\ov{v}_{\beta}', \ov{\omega}_{\beta}', \ov{\vartheta})=\varphi(v, \omega, \vartheta)+\varphi(\ov{v}, \ov{\omega}, \ov{\vartheta})
\end{equation}
for every $V=[v, \ov{v}, \omega, \ov{\omega}]\in\mathbb{R}^{6}$, $\beta=(\vartheta, \ov{\vartheta}, \psi)\in\mathbb{T}^{3}$, where
\begin{displaymath}
v_{\beta}':=\left(
\begin{array}{c}
\sigma_{\beta}[V]_{1} \\ \sigma_{\beta}[V]_{2}
\end{array}
\right), \quad \ov{v}_{\beta}':=\left(
\begin{array}{c}
\sigma_{\beta}[V]_{3} \\ \sigma_{\beta}[V]_{4}
\end{array}
\right), \quad \omega_{\beta}':=\sigma_{\beta}[V]_{5},\quad \ov{\omega}_{\beta}':=\sigma_{\beta}[V]_{6}.
\end{displaymath}
\end{defn}
We also make one more definition.
\begin{defn}
We define $\mathcal{P}(\mathbb{Z}_{2}^{2})$ to be the class of reference particles $\mathsf{P}_{\ast}\subset\mathbb{R}^{2}$ which have reflection symmetries in the two canonical orthogonal axes of $\mathbb{R}^{2}$.
\end{defn}
We are now ready to state in precise terms the main result of this article.
\begin{thm}[Characterisation of Collision Invariants]\label{awesome}
Suppose $\mathsf{P}_{\ast}\in\mathcal{P}(\mathbb{Z}_{2}^{2})$ has the property that $\partial\mathsf{P}_{\ast}$ is analytic, i.e. of class $C^{\omega}$. Let $\mathcal{S}$ be the associated family of matrices
\begin{equation*}
\{M^{-1}(I-2\widehat{\gamma}_{\beta}\otimes \widehat{\gamma}_{\beta})M\}_{\beta\in\mathbb{T}^{3}} \quad \text{or}\quad  \{M^{-1}(I-2\widehat{\eta}_{\beta}\otimes \widehat{\eta}_{\beta})M\}_{\beta\in\mathbb{T}^{3}}.
\end{equation*}
If a measurable map is an $\mathcal{S}$-collision invariant, then it is necessarily of the form
\begin{equation*}
\varphi(v, \omega, \vartheta)=a(\vartheta)+b\cdot v+c\left(m|v|^{2}+J\omega^{2}\right),
\end{equation*}
for some $b_{1}, b_{2}, c\in\mathbb{R}$ and some measurable $a:\mathbb{S}^{1}\rightarrow\mathbb{R}$.
\end{thm}
\begin{rem}
For the proof of this theorem, we need only restrict our attention to the family of maps $\{\sigma_{\beta}\}_{\beta\in\mathbb{T}^{3}}$ with $\sigma_{\beta}=M^{-1}(I-2\widehat{\gamma}_{\beta}\otimes \widehat{\gamma}_{\beta})M$, as the proof for the family $\sigma_{\beta}=M^{-1}(I-2\widehat{\eta}_{\beta}\otimes \widehat{\eta}_{\beta})M$ follows directly thereafter. This will become clear in section \ref{helpmealg} below.
\end{rem}
\subsection{Rewriting the Functional Identity}
It will prove useful to rewrite the functional identity \eqref{collinv} in a way that allows us to employ properties of reflection matrices, to which each scattering matrix $\sigma_{\beta}=M^{-1}(I-2\widehat{\gamma}_{\beta}\otimes\widehat{\gamma}_{\beta})M$ is conjugate. Indeed, given a collision invariant $\varphi$, we define 
\begin{equation*}
\varphi_{0}(v, \omega, \vartheta):=\varphi(v, \omega, \vartheta)-\varphi(0, 0, \vartheta),
\end{equation*}
together with an associated map $\Phi_{\varphi}:\mathbb{R}^{6}\times \mathbb{T}^{2}\rightarrow\mathbb{R}$ by
\begin{equation*}
\Phi_{\varphi}(V; \vartheta, \ov{\vartheta}):=\varphi_{0}(v, \omega, \vartheta)+\varphi_{0}(\ov{v}, \ov{\omega}, \ov{\vartheta}),
\end{equation*}
noting that $\Phi_{\varphi}(0; \vartheta, \ov{\vartheta})=0$ for all $(\vartheta, \ov{\vartheta})\in\mathbb{T}^{2}$. It follows that $\varphi$ is a collision invariant if and only if  
\begin{equation}\label{der}
\Phi_{\varphi}(\sigma_{\beta}V; \vartheta, \ov{\vartheta})=\Phi_{\varphi}(V; \vartheta, \ov{\vartheta})
\end{equation}
for all $\beta\in\mathbb{T}^{3}$. Setting $P:=MV$, and also define the new map $\Phi_{\varphi}^{\ast}:\mathbb{R}^{6}\times\mathbb{T}^{2}\rightarrow\mathbb{R}$ by 
\begin{equation*}
\Phi_{\varphi}^{\ast}(V; \vartheta, \ov{\vartheta}):=\Phi_{\varphi}(M^{-1}V; \vartheta, \ov{\vartheta}),
\end{equation*}
we find that $\varphi$ is a collision invariant if and only if
\begin{equation}\label{reflect}
\Phi_{\varphi}^{\ast}([I-2\widehat{\gamma}_{\beta}\otimes \widehat{\gamma}_{\beta}]P; \vartheta, \ov{\vartheta})=\Phi_{\varphi}^{\ast}(P; \vartheta, \ov{\vartheta})
\end{equation}
for all $\beta\in\mathbb{T}^{3}$ and $P\in\mathbb{R}^{6}$.
It is now we make the important observation that if the orientations $(\vartheta, \ov{\vartheta})$ are {\em fixed}, then \eqref{reflect} implies that
\begin{equation*}
\Phi_{\varphi}^{\ast}(\cdot; \vartheta, \ov{\vartheta}) \quad \text{is constant on the group orbits}\hspace{2mm} G^{\ov{\vartheta}}_{\vartheta}P,
\end{equation*}
for any chosen $P\in\mathbb{R}^{6}$, where $G^{\ov{\vartheta}}_{\vartheta}\subseteq\mathrm{O}(6)$ is the subgroup generated by the 1-parameter family of reflection matrices $\{I-2\widehat{\gamma}_{\beta}\otimes\widehat{\gamma}_{\beta}\,:\,\psi\in\mathbb{S}^{1}\}$, namely
\begin{equation*}
G^{\ov{\vartheta}}_{\vartheta}:=\left\langle\left\{I-2\widehat{\gamma}_{\beta}\otimes\widehat{\gamma}_{\beta}\,:\,\psi\in\mathbb{S}^{1}\right\}\right\rangle.
\end{equation*}
Transforming back to $V$-variables and observing identity \eqref{der}, we obtain the following result.
\begin{prop}\label{getme}
The map $\varphi$ is a collision invariant if and only if for each $(\vartheta, \ov{\vartheta})\in\mathbb{T}^{2}$, the map $\Phi_{\varphi}(\cdot; \vartheta, \ov{\vartheta})$ is constant on the group orbits $\mathcal{G}^{\ov{\vartheta}}_{\vartheta}V\subset\mathbb{R}^{6}$ for every $V\in\mathbb{R}^{6}$, where $\mathcal{G}^{\ov{\vartheta}}_{\vartheta}:=M^{-1}G^{\ov{\vartheta}}_{\vartheta}M$.
\end{prop}
The basic problem is now to characterise the orbits of every point in $\mathbb{R}^{6}$ under the action of $G^{\ov{\vartheta}}_{\vartheta}$ for each $(\vartheta, \ov{\vartheta})\in\mathbb{T}^{2}$. This leads us to the concept of {\em energy-momentum} submanifolds of $\mathbb{R}^{6}$, which we introduce now.
\subsection{Energy-momentum Submanifolds of $\mathbb{R}^{6}$}
We recall from section \ref{algcon} above that the scattering matrices $\sigma_{\beta}$ under study conserve total kinetic energy, i.e.
\begin{equation*}
|M\sigma_{\beta}V|^{2}=|MV|^{2},
\end{equation*}
along with total linear momentum of the particles,
\begin{displaymath}
m\left(
\begin{array}{c}
(\sigma_{\beta}V)_{1} \\ (\sigma_{\beta}V)_{2}
\end{array}
\right)+m\left(
\begin{array}{c}
(\sigma_{\beta}V)_{3} \\ (\sigma_{\beta}V)_{4}
\end{array}
\right) = m\left(
\begin{array}{c}
V_{1} \\ V_{2}
\end{array}
\right)+
m\left(
\begin{array}{c}
V_{3} \\ V_{4}
\end{array}
\right),
\end{displaymath}
for all $\beta\in\mathbb{T}^{3}$, once $V\in\mathbb{R}^{6}$ has been prescribed. Writing these in the language of the previous section, we have that
\begin{equation*}
Y\mapsto |MY|^{2} \quad\text{is constant on the group orbits}\hspace{2mm}\mathcal{G}^{\ov{\vartheta}}_{\vartheta}V
\end{equation*}
together with
\begin{equation*}
Y\mapsto Y_{1}+Y_{3} \hspace{2mm} \text{and}\hspace{2mm} Y\mapsto Y_{2}+Y_{4} \hspace{2mm} \text{are constant on the group orbits}\hspace{2mm}\mathcal{G}^{\ov{\vartheta}}_{\vartheta}V,
\end{equation*}
for $Y\in\mathbb{R}^{6}$. On the basis of these observations for the family of scattering matrices 
\begin{equation*}
\{M^{-1}(I-2\widehat{\gamma}_{\beta}\otimes\widehat{\gamma}_{\beta})M\}_{\beta\in\mathbb{T}^{3}},
\end{equation*}
it seems reasonable to postulate that the orbits $\mathcal{G}^{\ov{\vartheta}}_{\vartheta}V$ are simply those subsets of $\mathbb{R}^{6}$ which are realised as the intersection of energy ellipsoids
\begin{equation*}
\mathcal{E}(V):=\left\{Y\in\mathbb{R}^{6}\,:\,|MY|^{2}=|MV|^{2}\right\}
\end{equation*}
with momentum planes 
\begin{equation*}
\Pi_{1}(V):=\left\{Y\in \mathbb{R}^{6}\,:\,Y_{1}+Y_{3}=V_{1}+V_{3}\right\}\quad\text{and}\quad \Pi_{2}(V):=\left\{Y\in \mathbb{R}^{6}\,:\, Y_{2}+Y_{4}=V_{2}+V_{4}\right\}.
\end{equation*}
Indeed, this is what we prove in section \ref{helpmealg} below by using group-theoretic techniques and a careful analysis of properties of the collision normals $\gamma_{\beta}$.

Firstly, let us put the previous argument in precise terms. For a given energy $\mathsf{e}>0$ and momentum vector $\mathsf{p}\in\mathbb{R}^{2}$ satisfying $\mathsf{e}^{2}> |\mathsf{p}|^{2}/2m$, we define the associated {\em energy-momentum manifold} $\mathsf{M}(\mathsf{e}, \mathsf{p})\subset\mathbb{R}^{6}$ by
\begin{equation*}
\mathsf{M}(\mathsf{e}, \mathsf{p}):=\left\{Y\in\mathbb{R}^{6}\,:\,|MY|=\mathsf{e}\hspace{2mm}\text{and}\hspace{2mm} \left(
\begin{array}{c}
Y_{1}+Y_{3} \\ Y_{2}+Y_{4}
\end{array}\right)=\frac{\mathsf{p}}{m}\right\}.
\end{equation*}
It is now our aim to show that if $V\in\mathbb{R}^{6}$ is arbitrary, and we denote
\begin{displaymath}
\mathsf{e}^{2}=|MV|^{2} \quad \text{and}\quad \mathsf{p}=m\left(
\begin{array}{c}
V_{1}+V_{3} \\
V_{2}+V_{4}
\end{array}\right),
\end{displaymath}
then the group orbits of points $V\in\mathbb{R}^{6}$ are given by 
\begin{displaymath}
\mathcal{G}^{\ov{\vartheta}}_{\vartheta}V=\left\{
\begin{array}{ll}
\mathsf{M}(\mathsf{e}, \mathsf{p}) & \quad \text{if}\hspace{2mm}\mathsf{e}^{2}>\frac{|\mathsf{p}|^{2}}{2m} \vspace{2mm}\\
\left\{\left[\frac{\mathsf{p}}{2m}, \frac{\mathsf{p}}{2m}, 0, 0\right]\right\} & \quad \text{if} \hspace{2mm}\mathsf{e}^{2}=\frac{|\mathsf{p}|^{2}}{2m},
\end{array}
\right.
\end{displaymath}
for {\em any} choice of orientations $(\vartheta, \ov{\vartheta})\in\mathbb{T}^{2}$, i.e. the group orbits are independent of the choice of particle orientations. In other words, we want to show the restriction map $\Phi(\cdot; \vartheta, \ov{\vartheta})|_{\mathsf{M}(\mathsf{e}, \mathsf{p})}$ is a constant function for all suitable $\mathsf{e}>0$ and $\mathsf{p}\in\mathbb{R}^{2}$ by identity \eqref{der}. Since we have nothing to show in the case that $\mathcal{G}^{\ov{\vartheta}}_{\vartheta}V$ is a singleton set, we assume henceforth that $\mathsf{e}^{2}>|\mathsf{p}|^{2}/2m$. 
As the energy-momentum submanifolds are homeomorphic to the 3-sphere, one can expect to reduce the study of the subgroup $\mathcal{G}^{\ov{\vartheta}}_{\vartheta}\subseteq\mathrm{O}(6)$ acting on $\mathsf{M}(\mathsf{e}, \mathsf{p})$ to one of a group acting on $\mathbb{S}^{3}$. As done before in section \ref{diskcanon} above, let us now reduce our problem to a kind of canonical form.
\subsection{Transformation to Canonical Form}\label{trancon}
Let both energy $\mathsf{e}$ and momentum $\mathsf{p}$ be given which satisfy $\mathsf{e}^{2}>|\mathsf{p}|^{2}/2m$, and suppose them to be {\em fixed}. We now define $h_{\mathsf{e}, \mathsf{p}}:\mathsf{M}(\mathsf{e}, \mathsf{p})\rightarrow\mathbb{S}^{3}$ by
\begin{displaymath}
h_{\mathsf{e}, \mathsf{p}}[V]:=\frac{1}{r(V)}\left(
\begin{array}{c}
(MV)_{1}-(MV)_{3} \\
(MV)_{2}-(MV)_{4} \\
\sqrt{2}(MV)_{5} \\
\sqrt{2}(MV)_{6}
\end{array}
\right),
\end{displaymath}
where $r(V):=\sqrt{((MV)_{1}-(MV)_{3})^{2}+((MV)_{2}-(MV)_{4})^{2}+2(MV)_{5}^{2}+2(MV)_{6}^{2}}$, thereby considering $\mathbb{S}^{3}$ as embedded in $\mathbb{R}^{4}$. Notice also that since $\mathsf{e}^{2}>|\mathsf{p}|^{2}/2m$, the radicand of $r(V)$ is strictly positive. One can check that $h_{\mathsf{e}, \mathsf{p}}$ is a bijection between $\mathsf{M}(\mathsf{e}, \mathsf{p})$ and $\mathbb{S}^{3}$, 
whose inverse is given explicitly by
\begin{displaymath}
h_{\mathsf{e}, \mathsf{p}}^{-1}[w]=\frac{1}{2}\left(
\begin{array}{c}
\sqrt{2\mathsf{e}^{2}-\frac{|\mathsf{p}|^{2}}{m}}w_{1}+\frac{\mathsf{p}_{1}}{\sqrt{m}}\\
\sqrt{2\mathsf{e}^{2}-\frac{|\mathsf{p}|^{2}}{m}}w_{2}+\frac{\mathsf{p}_{2}}{\sqrt{m}}\\
\frac{\mathsf{p}_{1}}{\sqrt{m}}-\sqrt{2\mathsf{e}^{2}-\frac{|\mathsf{p}|^{2}}{m}}w_{1}\\
\frac{\mathsf{p}_{2}}{\sqrt{m}}-\sqrt{2\mathsf{e}^{2}-\frac{|\mathsf{p}|^{2}}{m}}w_{2}\\
\sqrt{\frac{\mathsf{e}^{2}}{2}-\frac{|\mathsf{p}|^{2}}{4m}}w_{3}\\
\sqrt{\frac{\mathsf{e}^{2}}{2}-\frac{|\mathsf{p}|^{2}}{4m}}w_{4}
\end{array}
\right) \quad \text{for}\hspace{2mm} w=(w_{1}, w_{2}, w_{3}, w_{4})\in\mathbb{S}^{3}.
\end{displaymath} 
We now consider the orbits $\mathcal{G}^{\ov{\vartheta}}_{\vartheta}V$ as images of another group action on $\mathbb{S}^{3}$ under the map $h_{\mathsf{e}, \mathsf{p}}$. A quick calculation reveals that
\begin{equation*}
\sigma_{\beta}\,:\,V\mapsto M^{-1}\left(I-2\widehat{\gamma}_{\beta}\otimes \widehat{\gamma}_{\beta}\right)MV \quad \text{for} \hspace{2mm}V\in\mathsf{M}(\mathsf{e}, \mathsf{p})\hspace{2mm}
\end{equation*}
if and only if
\begin{equation*}
s_{\beta}\,:\,w\mapsto\left(I-2\widehat{\mu}_{\beta}\otimes \widehat{\mu}_{\beta}\right)w \quad \text{for}\hspace{2mm}w=h_{\mathsf{e}, \mathsf{p}}(V),
\end{equation*}
where $\widehat{\mu}_{\beta}\in\mathbb{S}^{3}$ is the unit vector
\begin{displaymath}
\widehat{\mu}_{\beta}:=\sqrt{\frac{2}{\Lambda_{\beta}}}M^{-1}_{1}\left(
\begin{array}{c}
(N_{\beta})_{1}\\
(N_{\beta})_{2}\\
\frac{1}{\sqrt{2}}\left(r_{\beta}-d_{\beta}e(\psi)\right)^{\perp}\cdot N_{\beta}\\
-\frac{1}{\sqrt{2}}r_{\beta}^{\perp}\cdot N_{\beta}
\end{array}
\right),
\end{displaymath}
$\Lambda_{\beta}>0$ is given in \eqref{lambdaish} above, and the reduced mass-inertia matrix $M_{1}\in\mathbb{R}^{4\times 4}$ is given by
\begin{displaymath}
M_{1}:=\left(
\begin{array}{cccc}
\sqrt{m} & 0 & 0 & 0 \\
0 & \sqrt{m} & 0 & 0 \\
0 & 0 & \sqrt{J} & 0 \\
0 & 0 & 0 & \sqrt{J}
\end{array}
\right).
\end{displaymath}
It will be crucial for the proof of characterisation of collision invariants in the sequel to show that the $(\vartheta, \ov{\vartheta})$-dependent family of unit vectors $\{\widehat{\mu}_{\beta}\,:\,\psi\in\mathbb{S}^{1}\}$ lies in no single hyperplane in $\mathbb{R}^{4}$. Indeed, we address this problem in proposition \ref{spanny} below. With this observation that we may essentially work on the sphere $\mathbb{S}^{3}$ for any pair of orientations $(\vartheta, \ov{\vartheta})\in\mathbb{T}^{2}$, we define the group $H^{\ov{\vartheta}}_{\vartheta}\subseteq\mathrm{O}(4)$ by
\begin{equation*}
H^{\ov{\vartheta}}_{\vartheta}:=\left\langle\left\{I-2\widehat{\mu}_{\beta}\otimes \widehat{\mu}_{\beta}\,:\,\psi\in\mathbb{S}^{1}\right\}\right\rangle,
\end{equation*}
which is now the primary object of study. We have the following proposition, which crystalises the above discussion.
\begin{prop}\label{transfertransitive}
Let $(\vartheta, \ov{\vartheta})\in\mathbb{T}^{2}$ be given. The group $H^{\ov{\vartheta}}_{\vartheta}\subseteq\mathrm{O}(4)$ acts transitively on $\mathbb{S}^{3}$ if and only if $\mathcal{G}^{\ov{\vartheta}}_{\vartheta}\subseteq\mathrm{O}(6)$ acts transitively on $\mathsf{M}(\mathsf{e}, \mathsf{p})$ for any {\em single} pair $(\mathsf{e}, \mathsf{p})$ satisfying $\mathsf{e}^{2}>|\mathsf{p}|^{2}/2m$.
\end{prop}
If the orbits under $\mathcal{G}^{\ov{\vartheta}}_{\vartheta}$ of any given point in $\mathbb{R}^{6}$ is indeed the corresponding energy-momentum manifold, we may immediately infer the existence of another measurable function $\widetilde{\Phi_{\varphi}}:\mathbb{R}^{2}\times \mathbb{R}\rightarrow\mathbb{R}$ such that
\begin{equation}\label{reduct}
\Phi_{\varphi}(V; \vartheta, \ov{\vartheta})=\widetilde{\Phi_{\varphi}}(mv+m\ov{v}, m|v|^{2}+J\omega^{2}+m|\ov{v}|^{2}+J\ov{\omega}^{2}),
\end{equation}
for all $V\in\mathbb{R}^{6}$. To show that \eqref{reduct} holds for some $\widetilde{\Phi_{\varphi}}$, we employ some new results contained in the appendix of this article on generators of the rotation group $\mathrm{O}(4)$, which are due to C. Viterbo.

\subsection{The Transitive Group Action of $H^{\ov{\vartheta}}_{\vartheta}$ on $\mathbb{S}^{3}$}\label{helpmealg}
The key result is the following, whose proof can be found in \textsc{Appendix} \ref{algebra}.
\begin{thm}\label{claude}
Suppose that $\mu:\mathbb{S}^{1}\rightarrow\mathbb{S}^{3}$ is a continuous, non-constant map. Let $H\subseteq\mathrm{O}(4)$ denote the group 
\begin{equation*}
H:=\left\langle\left\{I-2\mu_{\psi}\otimes\mu_{\psi}\,:\,\psi\in\mathbb{S}^{1}\right\}\right\rangle.
\end{equation*}
Then $H$ acts transitively on $\mathbb{S}^{3}$ unless the image set $\{\mu_{\psi}\,:\,\psi\in\mathbb{S}^{1}\}$ is strictly contained in some hyperplane in $\mathbb{R}^{4}$.
\end{thm}
Using this result directly, we are able to prove that $H^{\ov{\vartheta}}_{\vartheta}$ defined above does indeed act transitively on $\mathbb{S}^{3}$. In fact, the proof of \textsc{Proposition} \ref{transfertransitive} follows immediately from the following result, which says that the image set $\{\widehat{\mu}_{\beta}\,:\,\psi\in\mathbb{S}^{1}\}$ cannot lie in any one fixed hyperplane for {\em any} choice of orientations $(\vartheta, \ov{\vartheta})\in\mathbb{T}^{2}$.
\begin{prop}\label{spanny}
For any $(\vartheta, \ov{\vartheta})\in\mathbb{T}^{2}$, we have $\mathrm{span}\{\widehat{\mu}_{\beta}\,:\,\psi\in\mathbb{S}^{1}\}=\mathbb{R}^{4}$.
\end{prop}
\begin{proof}
Let $(\vartheta, \ov{\vartheta})\in\mathbb{T}^{2}$ be given. We suppose, for a contradiction, that there exists a vector $W(\vartheta, \ov{\vartheta})\in \mathbb{R}^{4}\setminus\{0\}$, written componentwise as 
\begin{equation*}
W(\vartheta, \ov{\vartheta})=(w_{1}(\vartheta, \ov{\vartheta}), w_{2}(\vartheta, \ov{\vartheta}), w_{3}(\vartheta, \ov{\vartheta}), w_{4}(\vartheta, \ov{\vartheta})),
\end{equation*}
such that 
\begin{equation}\label{normass}
\widehat{\mu}_{\beta}\cdot W(\vartheta, \ov{\vartheta})=0 \quad \text{for all}\hspace{2mm} \psi\in\mathbb{S}^{1}.
\end{equation}
However, we note that this is equivalent to the statement that
\begin{equation*}
\gamma_{\beta}\cdot V(\vartheta, \ov{\vartheta})=0\quad\text{for all}\hspace{2mm}\psi\in\mathbb{S}^{1}, 
\end{equation*}
where
\begin{equation*}
V(\vartheta, \ov{\vartheta}):=\left(0, 0, -\sqrt{2/m}w_{1}, -\sqrt{2/m}w_{2}, \sqrt{1/J}w_{3}, \sqrt{1/J}w_{4}\right).
\end{equation*}
Importantly, assumption \eqref{normass} implies that
\begin{equation*}
\bigcap_{\psi\in\mathbb{S}^{1}}\Sigma_{\beta}^{0}\neq\{0\},
\end{equation*}
where $\Sigma_{\beta}^{0}=\Sigma_{\beta}^{-}\cap\Sigma_{\beta}^{+}$. In particular, there is at least one non-zero velocity vector $V(\vartheta, \ov{\vartheta})$ which is {\em both} pre- and post-collisional for {\em every} choice of elevation angle $\psi\in\mathbb{S}^{1}$. As the notion of pre- and post-collisional velocities is inherently dynamic, we must now appeal to the existence results established in section \ref{connydym}.

Let us consider the following 1-parameter family of initial data $Z_{0}(\psi)=[z_{0}, \ov{z}_{0}(\psi)]$ (parameterised by $\psi\in\mathbb{S}^{1}$) for the ODEs presented in section \ref{dynam}, where
\begin{equation*}
z_{0}=[0, \vartheta, 0, \omega] \quad \text{and} \quad \ov{z}_{0}(\psi)=[d^{\ov{\vartheta}}_{\vartheta}(\psi)e(\psi), \ov{\vartheta}, \ov{v}, \ov{\omega}].
\end{equation*}
with
\begin{displaymath}
\ov{v}:=-\sqrt{\frac{2}{m}}\left(
\begin{array}{c}
w_{1} \\ w_{2}
\end{array}
\right), \quad \omega:=-\sqrt{\frac{1}{J}}w_{3}\quad \text{and}\quad \ov{\omega}:=-\sqrt{\frac{1}{J}}w_{4}.
\end{displaymath}
It then follows that for the initial data $[z_{0}, \ov{z}_{0}(\psi)]$ and their associated phase trajectories $t\mapsto x_{\psi}(t)$, $t\mapsto\vartheta_{\psi}(t)$ and $t\mapsto\ov{x}_{\psi}(t)$, $t\mapsto \ov{\vartheta}_{\psi}(t)$ (which are smooth, by the results in \cite{ballard}) there exists $\delta>0$ independent of $\psi$ such that
\begin{equation}\label{psicond}
F(x_{\psi}(t), \ov{x}_{\psi}(t), \vartheta_{\psi}(t), \ov{\vartheta}_{\psi}(t))\geq 0 \quad \text{for all}\hspace{2mm}-\delta<t<\delta.
\end{equation}
We use this deduction to derive our contradiction by reducing our considerations to properties of the motion of the point of contact on particle $\ov{\mathsf{P}}$ both before and after collision. In the sequel, we often suppress the dependence on $\vartheta, \ov{\vartheta}$ for all relevant quantities of interest, in order to make the presentation of our arguments clearer.

We perform a time-dependent change of variables so that particle $\mathsf{P}$ is stationary for all time, and the dynamics of $\ov{\mathsf{P}}$ takes place in the exterior domain $\ov{\mathbb{R}^{2}\setminus\mathsf{P}}$. It will be convenient to take the view of material point trajectories which evolve on the particles $\mathsf{P}$ and $\ov{\mathsf{P}}$. Firstly, let $X_{\mathsf{P}}(t; x_{0})$ denote the position of the point on particle $\mathsf{P}$ at time $t\in\mathbb{R}$ whose initial position at time $t=0$ is $x_{0}$, namely
\begin{equation*}
X_{\mathsf{P}}(t; x_{0}):=R(\omega t)x_{0}
\end{equation*}
for any $x_{0}\in\mathsf{P}$. Similarly, let $X_{\ov{\mathsf{P}}}(t; x_{0})$ denote the position of the analogous point on particle $\ov{\mathsf{P}}$ at time $t\in\mathbb{R}$, i.e.
\begin{equation*}
X_{\ov{\mathsf{P}}}(t; x_{0}):=R(\ov{\omega}t)\left(x_{0}-d(\psi)e(\psi)\right)+d(\psi)e(\psi)+\ov{v}t,
\end{equation*}
for any $ x_{0}\in\ov{\mathsf{P}}$. Transforming to the time-dependent reference frame from which $\mathsf{P}$ is viewed as stationary, $X_{\mathsf{P}}(t; x_{0})\mapsto \widetilde{X_{\mathsf{P}}}(t, x_{0})$ and $X_{\ov{\mathsf{P}}}(t; x_{0})\mapsto \widetilde{X_{\ov{\mathsf{P}}}}(t; x_{0})$, where
\begin{equation*}
\widetilde{X_{\mathsf{P}}}(t; x_{0})=x_{0} \quad \text{for}\hspace{2mm}x_{0}\in R(\vartheta)\mathsf{P}_{\ast},
\end{equation*}
for all time $t\in\mathbb{R}$, and
\begin{equation*}
\widetilde{X_{\ov{\mathsf{P}}}}(t; x_{0})=R((\ov{\omega}-\omega)t)\left(x_{0}-d(\psi)e(\psi)\right)+R(-\omega t)\left(d(\psi)e(\psi)+\ov{v}t\right),
\end{equation*}
for $x_{0}\in R(\ov{\vartheta}(t))\mathsf{P}_{\ast}+\ov{v}(t)$. As such, we may conveniently view the motion of individual points on the particle $\ov{\mathsf{P}}$ as taking place in the exterior domain $\mathbb{R}^{2}\setminus R(\vartheta)\mathsf{P}_{\ast}$. In order to derive our contradiction, namely that $\cap_{\psi\in\mathbb{S}^{1}}\Sigma_{\beta}^{0}$ must indeed be the singleton $\{0\}$, we focus our attention on the trajectory of the point of collision which lies on particle $\ov{\mathsf{P}}(t)$. For the $C^{1}(-\delta, \delta)$ trajectory $t\mapsto \widetilde{X_{\ov{\mathsf{P}}}}(t; p(\psi))$ to satisfy
\begin{equation*}
\left\{\widetilde{X_{\ov{\mathsf{P}}}}(t; p(\psi))\,:\,t\in (-\delta, \delta)\right\}\subset \overline{\mathbb{R}^{2}\setminus R(\vartheta)\mathsf{P}_{\ast}} \quad\text{for all}\hspace{2mm}\psi\in\mathbb{S}^{1},
\end{equation*}
it is necessary that the normal component of the curve $\{\widetilde{X_{\ov{\mathsf{P}}}}(t; p(\psi))\,:\,t\in (-\delta, \delta)\}$ vanish at $t=0$, i.e.
\begin{equation*}
\frac{d}{dt}\widetilde{X_{\ov{\mathsf{P}}}}(t; p(\psi))\bigg|_{t=0}\cdot n(\psi)=0 \quad \text{for all} \hspace{2mm} \psi\in\mathbb{S}^{1}.
\end{equation*}
A calculation reveals that this holds if and only if 
\begin{equation}\label{contra}
\xi(\psi)\cdot W=0 \quad \text{for all}\hspace{2mm}\psi\in\mathbb{S}^{1},
\end{equation}
where $\xi=\xi^{\ov{\vartheta}}_{\vartheta}(\psi)\in\mathbb{R}^{4}$ is given by
\begin{displaymath}
\xi^{\ov{\vartheta}}_{\vartheta}(\psi):=\left[
\begin{array}{c}
n^{\ov{\vartheta}}_{\vartheta}(\psi) \\
-p^{\ov{\vartheta}}_{\vartheta}(\psi)^{\perp}\cdot n^{\ov{\vartheta}}_{\vartheta}(\psi) \\
\left(p^{\ov{\vartheta}}_{\vartheta}(\psi)-d^{\ov{\vartheta}}_{\vartheta}(\psi)e(\psi)\right)^{\perp}\cdot n^{\ov{\vartheta}}_{\vartheta}(\psi)
\end{array}
\right].
\end{displaymath}
We now show that the linear span of the set $\left\{\xi^{\ov{\vartheta}}_{\vartheta}(\psi)\,:\,\psi\in\mathbb{S}^{1}\right\}$ is the whole space $\mathbb{R}^{4}$ for any choice of $(\vartheta, \ov{\vartheta})\in\mathbb{T}^{2}$, which implies that $W\in\mathbb{R}^{4}$ must indeed be the zero vector by \eqref{contra} above. We require the result of the following simple lemma. 
\begin{lem}
Suppose $\mathsf{P}_{\ast}\in\mathcal{P}(\mathbb{Z}^{2}_{2})$. There exist at least two angles $\psi_{1}=\psi_{1}(\vartheta, \ov{\vartheta}), \psi_{2}=\psi_{2}(\vartheta, \ov{\vartheta})\in\mathbb{S}^{1}$ such that $p^{\ov{\vartheta}}_{\vartheta}(\psi_{i})^{\perp}\cdot n^{\ov{\vartheta}}_{\vartheta}(\psi_{i})=0$.
\end{lem}
\begin{proof}
We recall that one axis of symmetry of $\mathsf{P}_{\ast}$ lies along the $x$-axis, and the other lies along the $y$-axis. We denote by $\delta_{x}>0$ and $\delta_{y}>0$ the largest positive values of the $x$- and $y$-co-ordinates that lie on these axes of symmetry, respectively. Consider the angle $\psi_{1}=\psi_{1}(\vartheta, \ov{\vartheta})\in\mathbb{S}^{1}$ that gives rise to the point $p^{\ov{\vartheta}}_{\vartheta}(\psi_{1})=R(\vartheta)(\delta_{x}, 0)$ and the associated normal vector $n^{\ov{\vartheta}}_{\vartheta}(\psi_{1})$ to $\mathsf{P}$ at $p^{\ov{\vartheta}}_{\vartheta}(\psi_{1})$. Since the reference particle $\mathsf{P}_{\ast}$ has $\mathbb{Z}_{2}\times\mathbb{Z}_{2}$ symmetry, it follows that $R(\vartheta)K_{1}R(\vartheta)^{T}\mathsf{P}=\mathsf{P}$. Moreover, as $\partial\mathsf{P}_{\ast}$ is of class $C^{\omega}$ and so the outward normal at $p^{\ov{\vartheta}}_{\vartheta}(\psi_{1})$ is unique, it follows that $n^{\ov{\vartheta}}_{\vartheta}(\psi_{1})=R(\vartheta)(1, 0)$, whence $p^{\ov{\vartheta}}_{\vartheta}(\psi_{1})^{\perp}\cdot n^{\ov{\vartheta}}_{\vartheta}(\psi_{1})=0$. The other case follows by considering $\psi_{2}=\vartheta+\pi/2$, and arguing similarly by using the fact that $R(\vartheta)K_{2}R(\vartheta)^{T}\mathsf{P}=\mathsf{P}$.
\end{proof}
We now make the following four judicious choices of the angle of elevation $\psi\in\mathbb{S}^{1}$ to produce vectors $\{\xi_{1}, \xi_{2}, \xi_{3}, \xi_{4}\}$ which are candidates for a basis. Using the result of the above lemma, we choose $\psi_{1}\in\mathbb{S}^{1}$ with the property that $p^{\ov{\vartheta}}_{\vartheta}(\psi_{1})=R(\vartheta)(\delta_{x}, 0)$ and $p^{\ov{\vartheta}}_{\vartheta}(\psi_{1})^{\perp}\cdot n^{\ov{\vartheta}}_{\vartheta}(\psi_{1})=0$, which yields the vector $\xi_{1}:=\xi^{\ov{\vartheta}}_{\vartheta}(\psi_{1})$ given by
\begin{displaymath}
\xi_{1}=\mathsf{Q}_{\vartheta}\left(
\begin{array}{c}
1 \\
0 \\
0 \\
d^{\ov{\vartheta}}_{\vartheta}(\psi_{1})\sin\psi_{1}
\end{array}
\right),
\end{displaymath}
where $\mathsf{Q}_{\vartheta}\in\mathrm{O}(4)$ is the rotation matrix
\begin{displaymath}
\mathsf{Q}_{\vartheta}:=\left(
\begin{array}{cccc}
\cos\vartheta & -\sin\vartheta & 0 & 0 \\
\sin\vartheta & \cos\vartheta & 0 & 0 \\
0 & 0 & 1 & 0 \\
0 & 0 & 0 & 1
\end{array}
\right).
\end{displaymath}
Choosing $\psi_{2}=\psi_{1}+\pi/2$ and following similar reasoning, we yield $\xi_{2}=\xi^{\ov{\vartheta}}_{\vartheta}(\psi_{2})$ given by
\begin{displaymath}
\xi_{2}=\mathsf{Q}_{\vartheta}\left(
\begin{array}{c}
0 \\
1 \\
0 \\
d^{\ov{\vartheta}}_{\vartheta}(\psi_{1}+\frac{\pi}{2})\sin\psi_{1}
\end{array}
\right).
\end{displaymath}
Next, we choose any $\psi_{3}$ with $\psi_{1}<\psi_{3}<\psi_{2}$ satisfying the property that
\begin{displaymath}
\left(
\begin{array}{c}
d^{\ov{\vartheta}}_{\vartheta}(\psi_{1})\sin\psi_{1} \\
d^{\ov{\vartheta}}_{\vartheta}(\psi_{1}+\frac{\pi}{2})\sin\psi_{1}
\end{array}
\right)\cdot n_{\ov{\vartheta}-\vartheta}(\psi_{3}-\vartheta)\neq 0,
\end{displaymath}
together with $p^{\ov{\vartheta}}_{\vartheta}(\psi_{3})^{\perp}\cdot n^{\ov{\vartheta}}_{\vartheta}(\psi_{3})\neq 0$; we note that this is always possible since $p^{\ov{\vartheta}}_{\vartheta}(\psi)^{\perp}\cdot n^{\ov{\vartheta}}_{\vartheta}(\psi)=0$ for all $\psi$ satisfying $\psi_{1}<\psi<\psi_{2}+\pi/2$ if and only if $\mathsf{P}_{\ast}$ is a disk. Indeed, for such a $\psi_{3}\in\mathbb{S}^{1}$, we set $\xi_{3}:=\xi^{\ov{\vartheta}}_{\vartheta}(\psi_{3})$, where
\begin{displaymath}
\xi_{3}:=\mathsf{Q}_{\vartheta}\left(
\begin{array}{c}
n_{\ov{\vartheta}-\vartheta}(\psi_{3}-\vartheta)_{1} \\
n_{\ov{\vartheta}-\vartheta}(\psi_{3}-\vartheta) _{2} \\
-p^{\ov{\vartheta}}_{\vartheta}(\psi_{3})^{\perp}\cdot n^{\ov{\vartheta}}_{\vartheta}(\psi_{3}) \\
\left(p^{\ov{\vartheta}}_{\vartheta}(\psi_{3}) -d^{\ov{\vartheta}}_{\vartheta}(\psi_{3}) e(\psi_{3})\right)\cdot n^{\ov{\vartheta}}_{\vartheta}(\psi_{3}) 
\end{array}
\right).
\end{displaymath}
Finally, we choose $\psi_{4}=\psi_{3}+\pi$ and set $\xi_{4}=\xi^{\ov{\vartheta}}_{\vartheta}(\psi_{4})$, which yields by symmetry that
\begin{displaymath}
\xi_{4}:=\mathsf{Q}_{\vartheta}\left(
\begin{array}{c}
-n_{\ov{\vartheta}-\vartheta}(\psi_{3}-\vartheta)_{1} \\
-n_{\ov{\vartheta}-\vartheta}(\psi_{3}-\vartheta) _{2} \\
-p^{\ov{\vartheta}}_{\vartheta}(\psi_{3})^{\perp}\cdot n^{\ov{\vartheta}}_{\vartheta}(\psi_{3}) \\
\left(p^{\ov{\vartheta}}_{\vartheta}(\psi_{3}) -d^{\ov{\vartheta}}_{\vartheta}(\psi_{3}) e(\psi_{3})\right)\cdot n^{\ov{\vartheta}}_{\vartheta}(\psi_{3}) 
\end{array}
\right).
\end{displaymath}
With these observations in place, we approach the following lemma.
\begin{lem}\label{basisone}
The set $\{\xi_{i}\}_{i=1}^{4}$ is a basis for $\mathbb{R}^{4}$ if and only if $\sin\psi_{1}\neq 0$.
\end{lem}
\begin{proof}
We need only show that $\{\xi_{1}', \xi_{2}', \xi_{3}', \xi_{4}'\}$ is a basis for $\mathbb{R}^{4}$ when $\sin\psi_{1}\neq 0$, where $\xi_{j}':=\mathsf{Q}_{\vartheta}^{T}\xi_{j}$. Evidently, $\{\xi_{1}', \xi_{2}', \xi_{3}'\}$ is a set of linearly independent vectors. Assume for the moment there exist constants $(c_{1}, c_{2}, c_{3})\in\mathbb{R}^{3}\setminus\{0\}$ such that
\begin{equation*}
\xi_{4}'=c_{1}\xi_{1}'+c_{2}\xi_{2}'+c_{3}\xi_{3}'.
\end{equation*}
By necessity, $c_{3}=1$, since $\psi_{3}\in\mathbb{S}^{1}$ was chosen so that $p^{\ov{\vartheta}}_{\vartheta}(\psi_{3})^{\perp}\cdot n^{\ov{\vartheta}}_{\vartheta}(\psi_{3})\neq 0$. This immediately yields that $c_{1}=-2n_{\ov{\vartheta}-\vartheta}(\psi_{3}-\vartheta)_{1}$ and $c_{2}=-2n_{\ov{\vartheta}-\vartheta}(\psi_{3}-\vartheta)_{2}$. However, with these values of constants $c_{i}$ it must be that
\begin{displaymath}
\left(
\begin{array}{c}
d^{\ov{\vartheta}}_{\vartheta}(\psi_{1})\sin\psi_{1} \\
d^{\ov{\vartheta}}_{\vartheta}(\psi_{1}+\frac{\pi}{2})\sin\psi_{1}
\end{array}
\right)\cdot n_{\ov{\vartheta}-\vartheta}(\psi_{3}-\vartheta)= 0,
\end{displaymath}
which contradicts the properties of the elevation angle $\psi_{3}$. Thus, $\xi_{4}'$ cannot be a linear combination of $\xi_{1}', \xi_{2}', \xi_{3}'$, and so the set $\{\xi_{i}'\}_{i=1}^{4}$ constitutes a basis for $\mathbb{R}^{4}$ in the case where $\sin\psi_{1}\neq 0$. 
\end{proof}
To conclude the proof of the proposition, we 	need to consider the construction of another basis in the case when $\sin\psi_{1}=0$. To do this, we consider the auxiliary function on $\mathbb{S}^{1}$ given by the rule
\begin{equation*}
\psi\mapsto \frac{p_{\ov{\vartheta}-\vartheta}(\psi)^{\perp}\cdot n_{\ov{\vartheta}-\vartheta}(\psi)}{d_{\ov{\vartheta}-\vartheta}(\psi)e(\psi)^{\perp}\cdot n_{\ov{\vartheta}-\vartheta}(\psi)}.
\end{equation*}
Notably, this function vanishes when $\psi=\psi_{1}$ or $\psi=\psi_{2}$. Importantly, the numerator and denominator are both identically zero for all $\psi\in\mathbb{S}^{1}$ if and only if $\mathsf{P}_{\ast}$ is a disk. Since, by assumption, $\mathsf{P}_{\ast}$ is not a disk and its boundary $\partial\mathsf{P}_{\ast}$ is $C^{\omega}$, this function is non-constant and smooth away from those points where the denominator vanishes. We therefore choose any two distinct $\psi_{3}, \psi_{4}\in\mathbb{S}^{1}$ with the property that $p^{\ov{\vartheta}}_{\vartheta}(\psi_{i})^{\perp}\cdot n^{\ov{\vartheta}}_{\vartheta}(\psi_{i})\neq 0$ for $i=3, 4$ and 
\begin{equation*}
\frac{p_{\ov{\vartheta}-\vartheta}(\psi_{3})^{\perp}\cdot n_{\ov{\vartheta}-\vartheta}(\psi_{3})}{d_{\ov{\vartheta}-\vartheta}(\psi_{3})e(\psi_{3})^{\perp}\cdot n_{\ov{\vartheta}-\vartheta}(\psi_{3})}\neq \frac{p_{\ov{\vartheta}-\vartheta}(\psi_{4})^{\perp}\cdot n_{\ov{\vartheta}-\vartheta}(\psi_{4})}{d_{\ov{\vartheta}-\vartheta}(\psi_{4})e(\psi_{4})^{\perp}\cdot n_{\ov{\vartheta}-\vartheta}(\psi_{4})}.
\end{equation*}
Using this observation, it follows from an argument identical to that found in the proof of lemma \ref{basisone} that the family $\{\xi_{1}, \xi_{2}, \widetilde{\xi_{3}}, \widetilde{\xi_{4}}\}$ constitutes a basis for $\mathbb{R}^{4}$, where
\begin{displaymath}
\widetilde{\xi_{3}}:=\mathsf{Q}_{\vartheta}\left(
\begin{array}{c}
n_{\ov{\vartheta}-\vartheta}(\psi_{3})_{1} \\
n_{\ov{\vartheta}-\vartheta}(\psi_{3})_{2} \\
-p_{\ov{\vartheta}-\vartheta}(\psi_{3})^{\perp}\cdot n_{\ov{\vartheta}-\vartheta}(\psi_{3}) \\
q_{\ov{\vartheta}-\vartheta}(\psi_{3})^{\perp}\cdot n_{\ov{\vartheta}-\vartheta}(\psi_{3})
\end{array}
\right), \quad \widetilde{\xi_{4}}:=\mathsf{Q}_{\vartheta}\left(
\begin{array}{c}
n_{\ov{\vartheta}-\vartheta}(\psi_{4})_{1} \\
n_{\ov{\vartheta}-\vartheta}(\psi_{4})_{2} \\
-p_{\ov{\vartheta}-\vartheta}(\psi_{4})^{\perp}\cdot n_{\ov{\vartheta}-\vartheta}(\psi_{4}) \\
q_{\ov{\vartheta}-\vartheta}(\psi_{4})^{\perp}\cdot n_{\ov{\vartheta}-\vartheta}(\psi_{4})
\end{array}
\right).
\end{displaymath}
Thus, we have shown that the span of the set $\{\widehat{\mu}_{\beta}(\psi)\,:\,\psi\in\mathbb{S}^{1}\}$ is indeed $\mathbb{R}^{4}$, which completes the proof of the proposition.
\end{proof}
We conclude by noticing that by \textsc{Theorem} \ref{claude} the group $H^{\ov{\vartheta}}_{\vartheta}$ acts transitively on $\mathbb{S}^{3}$ for every $\theta\in\mathbb{S}^{1}$, which immediately yields that $\mathcal{G}^{\ov{\vartheta}}_{\vartheta}$ acts transitively on energy momentum manifolds for {\em any} choice of orientation pair $(\vartheta, \ov{\vartheta})\in\mathbb{T}^{2}$. As a result, there exists a measurable map $\widetilde{\Phi_{\varphi}}$ such that
\begin{equation*}
\Phi_{\varphi}(V; \vartheta, \ov{\vartheta})=\widetilde{\Phi_{\varphi}}(mv+m\ov{v}, m|v|^{2}+J\omega^{2}+m|\ov{v}|^{2}+J\ov{\omega}^{2}).
\end{equation*}
We now prove that this representation formula implies that collision invariants $\varphi$ are necessarily of the form
\begin{equation*}
\varphi(v, \omega, \vartheta)=a(\vartheta)+b\cdot v+c\left(m|v|^{2}+J\omega^{2}\right),
\end{equation*}
for any constants $b_{1}, b_{2}, c\in\mathbb{R}$ and any function $a:\mathbb{S}^{1}\rightarrow\mathbb{R}$. To do this, we appeal to classical results on Cauchy's functional equation.
\begin{rem}
We believe that proposition \ref{spanny} holds true for an arbitrary compact, strictly convex reference particle $\mathsf{P}_{\ast}$ in $\mathbb{R}^{2}$ with $C^{\omega}$ boundary, although we have chosen not to explore this particular extension of proposition \ref{spanny}. 
\end{rem}
\begin{rem}
As one need not appeal to dynamical considerations in this case, the proof of proposition \ref{spanny} also holds for the family of matrices $\{M^{-1}(I-2\widehat{\eta}_{\beta}\otimes\widehat{\eta}_{\beta})M\}_{\beta\in\mathbb{T}^{3}}$ when the boundary curve $\partial\mathsf{P}_{\ast}$ of the associated reference particle $\partial\mathsf{P}_{\ast}$ is only of class $C^{1}$, as opposed to analytic.
\end{rem}
\subsection{Cauchy's Functional Equation}\label{itsdone}
The last remaining step in the proof of \textsc{Theorem} \ref{awesome} is proving the following proposition.
\begin{prop}\label{asscon}
Let $\mathsf{e}>0$ and $\mathsf{p}\in\mathbb{R}^{2}$ be such that $\mathsf{e}^{2}>|\mathsf{p}|^{2}/2m$, and let $\varphi$ be a collision invariant. Suppose that $\Phi_{\varphi}(\cdot; \vartheta, \ov{\vartheta})|_{\mathsf{M}(\mathsf{e}, \mathsf{p})}$ is a constant function. Then $\varphi$ is necessarily of the form
\begin{equation*}
\varphi(v, \omega, \vartheta)=a(\vartheta)+b\cdot v+c\left(m|v|^{2}+J\omega^{2}\right) \quad \text{for}\hspace{2mm}V\in\mathbb{R}^{6},
\end{equation*}
for constants $b_{1}, b_{2}, c\in\mathbb{R}$ and a function of orientation $a:\mathbb{S}^{1}\rightarrow\mathbb{R}$.
\end{prop}
\begin{proof}
The main idea of the proof is to transform identity \eqref{collinv} for collision invariants into Cauchy's well-known functional equation for a real-valued function $g$ on $\mathbb{R}$, namely
\begin{equation}\label{cfe}
g(x)+g(y)=g(x+y) \quad \text{for}\hspace{2mm} x, y\in\mathbb{R}.
\end{equation}
It is well known (see \textsc{Darboux} \cite{darb}) that under the assumption $g$ be continuous at a single point of $\mathbb{R}$, the only possible solutions of \eqref{cfe} are {\em linear} functions $g(z)=cz$, where $c\in\mathbb{R}$. Since we assume $\varphi$ to be measurable, Lusin's theorem immediately gives us enough continuity of $\varphi$ on $\mathbb{R}^{3}\times\mathbb{S}^{1}$ for the following arguments to be valid. Indeed, since \eqref{collinv} is equivalent to identity \eqref{der}, we notice that if $\Phi_{\varphi}(\cdot; \vartheta, \ov{\vartheta})|_{\mathsf{M}(\mathsf{e}, \mathsf{p})}$ is constant then $\Phi_{\varphi}$ is necessarily of the form
\begin{equation*}
\Phi_{\varphi}(V; \vartheta, \ov{\vartheta})=\widetilde{\Phi_{\varphi}}(mv+m\ov{v}, m|v|^{2}+J\omega^{2}+m|\ov{v}|^{2}+J\ov{\omega}^{2}; \vartheta, \ov{\vartheta})
\end{equation*}
for some measurable auxiliary function $\widetilde{\Phi_{\varphi}}$. Since it then holds by definition of $\Phi_{\varphi}$ that
\begin{equation*}
\widetilde{\Phi_{\varphi}}(mv+m\ov{v}, m|v|^{2}+J\omega^{2}+m|\ov{v}|^{2}+J\ov{\omega}^{2}; \vartheta, \ov{\vartheta})=\varphi_{0}(v, \omega, \vartheta)+\varphi_{0}(\ov{v}, \ov{\omega}, \ov{\vartheta}),
\end{equation*}
setting $\ov{v}=0$ and $\ov{\omega}=0$, we find that
\begin{equation*}
\varphi_{0}(v, \omega, \vartheta)=\widetilde{\Phi_{\varphi}}(mv, m|v|^{2}+J\omega^{2}; \vartheta, \ov{\vartheta}),
\end{equation*}
namely that the value of $\widetilde{\Phi_{\varphi}}$ is {\em independent} of its second parameter $\ov{\vartheta}$. By repeating this argument by instead setting $v=0$ and $\omega=0$, we conclude that $\widetilde{\Phi_{\varphi}}$ is independent of both $\vartheta$ and $\ov{\vartheta}$, namely that
\begin{equation}\label{almostdone}
\varphi_{0}(v, \omega, \vartheta)+\varphi_{0}(\ov{v}, \ov{\omega}, \ov{\vartheta})=\Psi_{\varphi}(v+\ov{v}, |v|^{2}+\textstyle\frac{J}{m}\omega^{2}+|\ov{v}|^{2}+\frac{J}{m}\ov{\omega}^{2})
\end{equation}
for some new measurable function $\Psi_{\varphi}$. Since $\varphi$ is assumed to be a collision invariant, if follows that $\Psi_{\varphi}$ satisfies the identity
\begin{equation*}
\textstyle\Psi_{\varphi}(v, |v|^{2}+\frac{J}{m}\omega^{2})+\Psi_{\varphi}(\ov{v}, |\ov{v}|^{2}+\frac{J}{m}\ov{\omega}^{2})=\Psi_{\varphi}(v+\ov{v}, |v|^{2}+\frac{J}{m}\omega^{2}+|\ov{v}|^{2}+\frac{J}{m}\ov{\omega}^{2}).
\end{equation*}
Finally, setting $\omega=\ov{\omega}=0$, we infer that 
\begin{equation}\label{psihelp}
\Psi_{\varphi}(v, |v|^{2})+\Psi_{\varphi}(\ov{v}, |\ov{v}|^{2})=\Psi_{\varphi}(v+\ov{v}, |v|^{2}+|\ov{v}|^{2}).
\end{equation}
It is at this point we invoke an argument from \textsc{Truesdell and Muncaster} \cite{MR554086}. Let us now make the choice $\ov{v}=-v$, which yields from \eqref{psihelp} that
\begin{equation}\label{morepsi}
\Psi_{\varphi}(0, 2|v|^{2})=\Psi_{\varphi}(v, |v|^{2})+\Psi_{\varphi}(-v, |v|^{2}).
\end{equation}
Next, selecting any two orthogonal vectors $v, \ov{v}$, we deduce from \eqref{psihelp} that
\begin{align}\label{lastpsi}
\Psi_{\varphi}(0, 2|v|^{2}+2|\ov{v}|^{2}) =  & \quad \Psi_{\varphi}(0, 2|v+\ov{v}|^{2}) \notag \\
\overset{\eqref{morepsi}}{=} & \quad \Psi_{\varphi}(v+\ov{v}, |v|^{2}+|\ov{v}|^{2})+\Psi_{\varphi}(-v-\ov{v}, |v|^{2}+|\ov{v}|^{2}) \notag \\
\overset{\eqref{psihelp}}{=} & \quad \Psi_{\varphi}(v, |v|^{2})+\Psi_{\varphi}(\ov{v}, |\ov{v}|^{2})+\Psi_{\varphi}(-v, |v|^{2})+\Psi_{\varphi}(-\ov{v}, |\ov{v}|^{2}) \notag \\
\overset{\eqref{morepsi}}{=} & \quad \Psi_{\varphi}(0, 2|v|^{2})+\Psi_{\varphi}(0, 2|\ov{v}|^{2}).
\end{align}
Thus, the map $g_{1}(s):=\Psi_{\varphi}(0, s)$ satisfies Cauchy's functional equation on $[0, \infty)$, and is therefore necessarily of the form $g_{1}(s)=cs$ for some $c\in\mathbb{R}$. Now consider the map $g_{2}(v):=\Psi_{\varphi}(v, |v|^{2}
)-g_{1}(|v|^{2})$. One may check that $g_{2}$ is measurable and odd on $\mathbb{R}^{2}$, and by \eqref{psihelp} and \eqref{lastpsi} above is additive on orthogonal pairs of vectors in $\mathbb{R}^{2}$. It follows from (\textsc{Truesdell and Muncaster} \cite{MR554086}, page 88) that $g_{2}$ is necessarily of the form $g_{2}(v)=b\cdot v$ for some $b\in\mathbb{R}^{2}$. As $\Psi_{\varphi}(v, |v|^{2})=g_{1}(|v|^{2})+g_{2}(v)$, it follows that
\begin{equation*}
\Psi_{\varphi}(v, |v|^{2})=b\cdot v+c|v|^{2}.
\end{equation*}
Thus, setting $\ov{v}=0$ and $\ov{\omega}=0$ in \eqref{almostdone} above, we deduce that $\varphi_{0}$ satisfies
\begin{equation*}
\varphi_{0}(v, \omega, \vartheta)=b\cdot v+c\left(m|v|^{2}+J\omega^{2}\right)
\end{equation*}
for some $b\in\mathbb{R}^{2}$ and $c\in\mathbb{R}$. Since any function of $\vartheta\in\mathbb{S}^{1}$ is a collision invariant, the claim of the proposition is proved.
\end{proof}

\subsection*{Acknowlegdements}
We would like to Claude Viterbo for his valuable contribution to this article. We would also like to thank Paul Seidel for comments related to the group theoretic arguments employed in this article. MW would like to thank the Fondation Sciences Math\'{e}matiques de Paris for its support of a Postdoctoral Fellowship which allowed him to carry out this research at the \'{E}cole Normale Sup\'{e}rieure de Paris, and would also like to thank Harvard University, where a part of this work was completed.

\appendix

\section{On Groups Generated by Reflections (by Claude Viterbo)}\label{algebra}
\def \Id {\rm Id}
We shall here prove the following result about the transitive group action of $H_{\vartheta}^{\ov{\vartheta}}$ on $\mathbb{S}^{3}$. Let $\mu: \mathbb{S}^1 \rightarrow \mathbb{S}^3$ be a continuous curve and $s:\mathbb{S}^{1}\rightarrow\mathrm{O}(4)$ be the associated hyperplane symmetries with respect to $\mu^\perp$, namely $s_{\psi}=I-2\mu_{\psi}\otimes \mu_{\psi}$ for $\psi\in\mathbb{S}^{1}$. 

\begin{prop} \label{Prop-1.1}
The group generated by the reflection matrices $\{s_{\psi}\,:\,\psi\in\mathbb{S}^{1}\}$ acts transitively on $\mathbb{S}^3$ unless the image of $\mu$ is contained in a hyperplane of $\mathbb{R}^4$. 
\end{prop}
It is important to mention that proposition \ref{Prop-1.1} extends the work of \textsc{Eaton and Perlman} (\cite{MR0463329}, theorem 1), in the sense that we do not need to take the Euclidean closure of $\langle\{s_{\psi}\,:\,\psi\in\mathbb{S}^{1}\}\rangle$ in order to infer that it is indeed the whole group $\mathrm{O}(4)$. In what follows, we actually prove the following more general result, from which \ref{Prop-1.1} follows.
\begin{prop} 
Let $\mu :A \rightarrow \mathbb{S}^{n-1}$ be a continuous map, where $A$ is connected and not reduced to a point. Let $s_{\psi}:=I-2\mu_{\psi}\otimes\mu_{\psi}$ be hyperplane symmetry matrices with respect to $\mu_{\psi}^\perp$. The group $G$ generated by $\{s_{\psi}\,:\,\psi\in A\}\subseteq\mathrm{O}(n)$ is identically equal to $\mathrm{O}(n)$ unless there is a $k$-dimensional hyperplane $\Pi\subset\mathbb{R}^{n}$ ($k\leq n-1$) such that $\mu_{\psi}\in \Pi$ for all $\psi\in A$
\end{prop} 

Note that if the image of $\mu$ is contained in a hyperplane $\Pi$, the orthogonal set $\Pi^\perp$ is invariant by all elements of the group generated by $\{s_{\psi}\,:\,\psi\in A\}\subseteq\mathrm{O}(n)$ and thus the associated action on $\mathbb{S}^{n-1}$ cannot be transitive.

The following result has been proved in \cite{MR0463329}: if the group $G$ is infinite, then its closure is equal to $\mathrm{O}(n)$. But since the map $s$ is non constant, the group generated by the elements $s_{\psi}$ is necessarily infinite. 
We may thus assume $G$ is dense in $\mathrm{O}(n)$. 

We note that the hyperplane symmetries $s_{\psi}$ have determinant $-1$. It will be useful to consider the group $K$, the intersection of $G$ with all proper rotations of 4-space $\mathrm{SO}(4)$. 
Since every element of $K$ can be written as the product of an even number of matrices $s_{\psi}$, we have the following: 

\begin{lem} The group $K$ is arcwise connected
\end{lem} 
\begin{proof}  As $\mathbb{S}^1$ is arcwise connected, we have $g=s_{\psi_{1}}s_{\psi_{2}}...s_{\psi_{2p-1}}s_{\psi_{2p}}$ is homotopic to $s_{1}s_{1}...s_{1}s_{1}=s_{1}^{2p}=I$ in $\mathrm{SO}(4)$ for any $g\in K$, where $1$ denotes the identity element of $\mathbb{S}^{1}$.
 \end{proof} 
 
We shall also need the following theorem.

 \begin{thm} [Kuranishi-Yamabe-Goto]\label{K-Y-G} Let $H$ be any connected subgroup of a Lie group $G$. Then $H$ is a Lie group. Moreover, there is a Lie subalgebra $\mathfrak h$ of  $\mathfrak g$ such that there exists a neighbourhood $V$ of the identity $e$ in $H$ and $U\subset \mathfrak h$ with $V= \exp(\mathfrak h\cap U)$
\end{thm} 
 \begin{proof} 
We refer the reader to \cite{MR0233923} (see also \cite{MR2047118} Theorem 11 p. 292, and p. 196). 
 \end{proof} 
Finally if $G$ is a Lie group and $H$ a connected subgroup, there is a (proper) maximal connected subgroup of $G$ containing $H$. We do not require Zorn's lemma, since we may simply take a subgroup of maximal dimension strictly less than $\dim(G)$ containing $H$. 
 
\begin{lem} 
A connected maximal subgroup of $\mathrm{SO}(n)$ is necessarily closed, hence compact. 
\end{lem} 
\begin{proof} A maximal subgroup is either closed or dense. 
We could use (\cite{MR1297165} Theorem 1.3, p. 628) applied to the special case of $\mathrm{SO}(n)$, which has the property the the connected component of its center is trivial. This result states the following:
if $G$ is a connected Lie group and $h: G \rightarrow \mathrm{SO}(n)$ is a Lie group homomorphism with dense image, then $h(G)=\mathrm{SO}(n)$. 
\end{proof} 
\begin{rem}
In the case $n\neq 4$ when the group $\mathrm{SO}(n)$ is simple, we have a simpler proof. Indeed, according to Theorem \ref{K-Y-G}, such a subgroup corresponds to a Lie algebra of $\mathrm{so}(n)$.
Let then $\mathfrak h$ be a Lie subalgebra of $\mathrm{so}(n)$ corresponding to a dense subgroup $H$. Since $\mathrm{Ad}(g)\mathfrak h =\mathfrak h$ for all $g\in H$, we have by density that this still holds for any $g\in \mathrm{SO}(n)$, hence $\mathfrak h$ is an ideal of $\mathrm{so}(n)$ and $H$ is a connected normal subgroup of $\mathrm{SO}(n)$. But this is impossible, since $\mathrm{so}(n)$ is a simple Lie algebra. 
\end{rem}
\begin{proof} [Proof of Proposition \ref{Prop-1.1}]
The group $K$ is dense, connected, and contained in a maximal connected subgroup which is of course dense. Thus $K=\mathrm{SO}(n)$. It is then follows at once that $G=\mathrm{O}(n)$. 
\end{proof}

\bibliography{biblio}

\end{document}